\documentclass[journal,twoside,web]{ieeecolor}
\usepackage{generic}
\usepackage{cite}
\usepackage{amsmath,amssymb,amsfonts}
\usepackage{algorithmic}
\usepackage{graphicx}
\usepackage{textcomp}
\usepackage{epstopdf}
\usepackage{accents}
\usepackage[ruled,vlined]{algorithm2e}
\usepackage{comment}

\newtheorem{proposition}{\textbf{Proposition}}
\newtheorem{problem}{\textbf{Problem}}
\newtheorem{corollary}{\textbf{Corollary}}
\newtheorem{remark}{\textbf{Remark}}
\newtheorem{theorem}{\textbf{Theorem}}
\newtheorem{definition}{\textbf{Definition}}
\newtheorem{lemma}{\textbf{Lemma}}
\newtheorem{assumption}{\textbf{Assumption}}

\newcommand{\ubar}{\underaccent{\bar}}

\def\BibTeX{{\rm B\kern-.05em{\sc i\kern-.025em b}\kern-.08em
    T\kern-.1667em\lower.7ex\hbox{E}\kern-.125emX}}
\begin{document}
\title{A Bayesian Nash equilibrium-based moving target defense against stealthy sensor attacks}
\author{David Umsonst, Serkan Sar{\i}ta\c{s}, György D\'an and Henrik Sandberg
\thanks{Date of Submission: February 25, 2021.
This work was supported in part by the Swedish Research Council (grants 2016-00861 and 2020-03860), and the Swedish Civil Contingencies Agency through the CERCES2.}
\thanks{David Umsonst is with Ericsson Research, Stockholm, Sweden
        {\tt\small david.umsonst@ericsson.com}}
\thanks{Henrik Sandberg is with the Division of Decision and Control Systems in the School of Electrical Engineering and Computer Science at the KTH Royal Institute of Technology, 10044 Stockholm, Sweden
        {\tt\small hsan@kth.se}}
\thanks{Serkan~Sar{\i}ta\c{s} is with the Department of Electrical and Electronics Engineering, Middle East Technical University, 06800, Ankara, Turkey.
        {\tt\small ssaritas@metu.edu.tr}}
\thanks{György D\'an is with the Division of Division of Network and Systems Engineering in the School of Electrical Engineering and Computer Science at the KTH Royal Institute of Technology, 10044 Stockholm, Sweden
        {\tt\small gyuri@kth.se}}}

\maketitle

\begin{abstract}
We present a moving target defense strategy to reduce the impact of stealthy sensor attacks on feedback systems.
The defender periodically and randomly switches between thresholds from a discrete set to increase the uncertainty for the attacker and make stealthy attacks detectable.
However, the defender does not know the exact goal of the attacker but only the prior of the possible attacker goals.
Here, we model one period with a constant threshold as a Bayesian game and use the Bayesian Nash equilibrium concept to find the distribution for the choice of the threshold in that period, which takes the defender's uncertainty about the attacker into account.
To obtain the equilibrium distribution, the defender minimizes its cost consisting of the cost for false alarms and the cost induced by the attack. 
We present a necessary and sufficient condition for the existence of a moving target defense and formulate a linear program to determine the moving target defense.
Furthermore, we present a closed-form solution for the special case when the defender knows the attacker's goals.
The results are numerically evaluated on a four-tank process.
\end{abstract}


\section{Introduction}
\label{sec:Introduction}
Critical infrastructures and control systems are increasingly connected to public communication networks, such as the Internet, and constitute geographically distributed cyber-physical systems (CPS). 
The use of public network infrastructures can save costs, for example cabling, but also increase the performance of the CPS.
However, this interconnection comes at the price of vulnerability to cyber-attacks, which are already impacting critical infrastructures, such as the Ukranian power grid~\cite{UkraineAttack}, as well as industrial control systems, such as a steel mill in Germany~\cite{GermanSteelAttack}.

To improve CPS security and complement existing information technology (IT) security measures, such as encryption and authentication, a new branch of security measures based on control-theoretic methods has emerged over the last decade. 
These novel security measures are based on the physics of the CPS and use physical models to detect, isolate, and mitigate malicious attacks. 
Hence, the control-theoretic security measures are a complementary approach to the IT security measures.
In the authors' opinion the main difference is between IT security measures and control-theoretic security measures is that IT security measures consider cryptography and logical isolation, while control-theoretic security measures are based on the physical models of the closed-loop system.
Using both control-theoretic and IT security measures constitutes a defense-in-depth approach to security.
An introduction to this topic can be found in the tutorial papers \cite{TutorialOnSecurityAnnualReview} and \cite{TutorialOnSecurityECC19}.
Since the attacker and the operator/defender are rational entities, their interaction is strategic and can thus be modeled using game-theoretic tools, see, for example, \cite{OverviewOfDynamicGamesInCPSSecurity}.

An emerging approach to detect attacks and to limit their impact, which can combine both physical models and game theory, are moving target defense (MTD) strategies \cite{MTDBook-PartI} that induce controlled uncertainties into the CPS to confuse the attacker.
Gairo~\emph{et al.}~\cite{MTDACC19SwitchingSensors}, for example, randomly switch between the sensors used to detect otherwise stealthy attacks, while in \cite{WatermarkingForNetworkedSystems} a random watermarking signal is injected into the CPS to make stealthy attacks detectable.
Furthermore, perturbations of power line impedances are analyzed in \cite{MTDInPowerSystem2020}, where an in-depth analysis is conducted to determine when the MTD will be successful.
Another system-switching approach is considered in \cite{KanellopoulosTAC2020}, which considers both actuator and sensor attacks.
However, the MTD strategies in \cite{MTDACC19SwitchingSensors, WatermarkingForNetworkedSystems,MTDInPowerSystem2020,KanellopoulosTAC2020} directly influence the closed-loop behavior of the CPS and can decrease its performance.
Griffioen \emph{et al.}~\cite{GriffioenTAC2021} propose three different MTD schemes, where the first one is similar to \cite{MTDACC19SwitchingSensors} but it also switches the plant and input matrix and not only the measurements used. The second MTD of \cite{GriffioenTAC2021} introduces an auxiliary system to not influence the closed-loop behavior and simultaneously detect an attack, while the third MTD utilizes the nonlinearities in the measurements.
In this work, we propose a moving target defense that is placed in the anomaly detector of the CPS, which is located outside of the control loop (see Figure~\ref{fig:BlockDiag}). 
Therefore, the proposed MTD neither influences the closed-loop performance directly, nor is there a need of introducing new auxiliary components to the system such that the controller and the moving target defense can be designed independently.

\subsection{Contribution}
When using an anomaly detector the defender faces a trade-off between the cost for false alarms and the cost for the impact of a stealthy attack, where ideally both costs should be as small as possible,
However, fewer false alarms typically lead to a larger attack impact, and vice versa, such that we cannot minimize both costs at the same time.
Therefore, we formulate a game where the defender periodically chooses a detector threshold at random to mitigate the trade-off between its cost for false alarms and its cost induced by the stealthy attack launched by the attacker.
The goal of the attacker is to maximize its payoff, which can, for example, be characterized by an unsafe region in the system's state space, while the defender wants to minimize the cost induced by false alarms and the cost induced by the attack.
However, the defender is uncertain about the payoff function the attacker tries to optimize and only has a belief of facing an attacker with a certain payoff function.
Here, we present an initial analysis of this game.
We consider a single period with a constant threshold and look at the threshold choice for this period. 
We show that there is an equivalent matrix game to analyse the equilibrium strategies of each player.

The matrix game formulation is used to provide a necessary and sufficient condition for when a Bayesian Nash equilibrium exists in which the defender's strategy is mixed and does not concentrate the whole probability on one action.
The defender's equilibrium strategy is then a moving target defense strategy.
Furthermore, by using the structure of the matrix game, we show that the Bayesian Nash equilibrium can be obtained by solving a linear program.
For the special case where the defender knows the attacker's type, we provide a closed-form solution for the Nash equilibrium, which gives us insights about the equilibrium strategies of the defender and attacker.
Finally, we numerically verify our results with a four-tank system.

\subsection{Related Work}
Since control systems are typically equipped with an anomaly detector to detect faults, several research groups have investigated how the choice and tuning of the anomaly detector threshold can help limiting the attack impact of stealthy attacks. 
When it comes to the choice of the detector, Murguia \emph{et al.}~\cite{RuthsMultivariate} compare a $\chi^2$ and a CUSUM detector and investigate which detector mitigates the impact of a sensor attack the most.

In the present work, we are interested in the case where the detector is already chosen and we want to define a way to choose the thresholds to limit the attack impact.
Urbina \emph{et al.}~\cite{Alvaro} point out that there will be a trade-off between the number of false alarms and the maximum impact of a stealthy attack when tuning the anomaly detector.

There are several other works that use the anomaly detector threshold to limit the attack impact or to detect attacks.
Ghafouri \emph{et al.}~\cite{GameTheoryThresholds2016} propose a Stackelberg game framework for choosing the detector threshold. 
Both a static choice as well as a dynamic choice of the detector threshold are presented, but the attack is assumed to be detectable. 
The cost that the defender wants to minimize is composed of the cost of false alarms, the cost of the attack impact, and the cost for switching between thresholds.

In \cite{UmsonstCDC18}, we extend the static detector threshold choice of \cite{GameTheoryThresholds2016} to the case of stealthy sensor attacks and prove the existence of such a threshold and provide conditions for the uniqueness.

Niu \emph{et al.}~\cite{DetectorThresholdSwitchingGameSec19} formulated the detector threshold switching problem as a zero-sum Stackelberg game without considering the cost for false alarms.

In \cite{UmsonstCDC20}, we consider a similar game but there the defender exactly knows the attacker's objective. This assumption is relaxed in the present work since the defender only has a prior over possible attacker objectives. Therefore, we extend the results of \cite{UmsonstCDC20} to a broader class of games, namely Bayesian games. Here, we also provide a closed-form solution to the special case considered in \cite{UmsonstCDC20}.

\subsection{Notation}
Let $x\in\mathbb{R}^n$ be an $n$-dimensional column vector and $A\in\mathbb{R}^{m\times n}$ be an $m$-by-$n$ matrix. 
The $i$th element of $x$ is denoted by $x_i$ and $A_{ij}$ corresponds to the element in the $i$th row and $j$th column of $A$. 
Further $x_{i:j}$ is the vector $[x_{i},\ x_{i+1}, \ldots, x_{j-1},\ x_{j}]^T$, where $i\leq j$.
If a random variable $x$ has a Gaussian distribution with mean $\mu\in\mathbb{R}^n$ and covariance matrix $\Sigma\in\mathbb{R}^{n\times n}$, we denote it as $x\sim\mathcal{N}(\mu,\Sigma)$.
The expected value of a random variable $x$ is denoted by $\mathbb{E}\lbrace x\rbrace$.
The $n$-by-$n$ identity matrix is denoted by $I_n$ and an $n$-dimensional column vector with all elements equal to one as $1_n$, while the indicator function of an event $D$ is represented by $\mathbf{1}_{\lbrace D\rbrace}$.
\section{System Model}
\label{sec:ProblemFormulation}
In this section, we introduce the models for the plant, controller, and detector and present our assumptions on the attacker and the defender.
Further, in Section~\ref{sec:DiscussionOfAssumptions} we will discuss the assumptions made on the system, the attacker, and the defender.
Figure~\ref{fig:BlockDiag} shows a block diagram of the sensor attack scenario that we consider.
\begin{figure}
    \centering
    \includegraphics[scale=0.4]{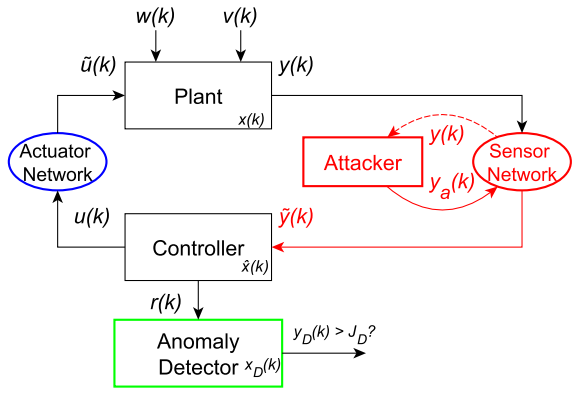}
    \caption{Block diagram of the sensor attack scenario.}
    \label{fig:BlockDiag}
\end{figure}
\subsection{Plant and Controller Model}
In our setup, the plant receives actuator signals and sends measurement signals over a network. 
We model the plant in Fig.~\ref{fig:BlockDiag} as a linear discrete-time system,
\begin{align*}
x(k+1)&=Ax(k)+B\tilde{u}(k)+w(k),\\
y(k)&=Cx(k)+v(k),
\end{align*}
where  $x(k)\in\mathbb{R}^{n_x}$ is the plant's state, $\tilde{u}(k)\in\mathbb{R}^{n_u}$ is the actuator signal received over the network, $y(k)\in\mathbb{R}^{n_y}$ is the measurement signal, $w(k)\in \mathbb{R}^{n_x}$ is the process noise, and $v(k)\in\mathbb{R}^{n_y}$ is the measurement noise. 
Both $w(k)$ and $v(k)$ have independent and identically distributed zero-mean multivariate Gaussian distributions with covariance matrices $\Sigma_w$ and $\Sigma_v$, respectively. Further, $w(k)$ and $v(k)$ are independent processes.
The system, input, and output matrices are $A\in\mathbb{R}^{n_x\times n_x}$, $B\in\mathbb{R}^{n_x\times n_u}$, and $C\in \mathbb{R}^{n_y\times n_x}$, respectively.
\begin{assumption}
\label{assum:StabilizableDetectableSystem}
$(A,B)$ is stabilizable, $(C,A)$ is detectable, and $(A,\Sigma_w^{\frac{1}{2}})$ has no uncontrollable modes on the unit circle.
\end{assumption}

The plant is controlled using a Kalman filter-based observer, which estimates the plant's state as $\hat{x}(k)\in\mathbb{R}^{n_x}$. The dynamics of the controller are
\begin{align*}
\hat{x}(k+1)&=A\hat{x}(k)+Bu(k)+L(\tilde{y}(k)-C\hat{x}(k)),\\
u(k)&=-K\hat{x}(k),
\end{align*}
where $\tilde{y}(k)$ is the measurement signal received over the network, $u(k)$ is the actuator signal determined by the controller, and $K$ and $L$ are the controller gain and steady-state Kalman gain, respectively.
Further, $L=APC^T(CPC^T+\Sigma_v)^{-1}$, where $P$ is the stabilizing solution to the algebraic Riccati equation
\begin{align*}
    P&=APA^T+\Sigma_w-(APC^T)(CPC+\Sigma_v)^{-1}(APC^T)^T,
\end{align*}
and $P$ exists due to Assumption~\ref{assum:StabilizableDetectableSystem}.

\subsection{Detector Model}
Since faults and/or malicious attacks can occur, the closed-loop system is equipped with an anomaly detector on the controller side, which has the possibly nonlinear dynamics
\begin{equation}
\label{eq:DetectorModel}
    \begin{aligned}
        x_D(k+1)&=\theta(x_D(k),r(k)),\\
        y_D(k+1)&=d(x_D(k),r(k)),
    \end{aligned}
\end{equation}
where $x_D(k)\in\mathbb{R}^{n_D}$ is the detector's internal state, ${y_D(k)\in\mathbb{R}_{\geq 0}}$ is the detector output, and $r(k)\in\mathbb{R}^{n_y}$ is the detector input. 
The exact structure of $\theta(x_D(k),r(k))$ and $d(x_D(k),r(k))$ depends on the detector the defender will use. 
For example, in Section~\ref{sec:NumericalEvaluation} we consider the static $\chi^2$ detector, i.e., $y_D(k)=\|r(k+1)\|_2^2$. More detector models can be found in \cite{FaultTolerantControl,ModelBasedFaultDiagnosis}.

We define the input $r(k)$ to be the residual signal, which is the normalized difference between the received and the predicted measurements, i.e.,
\begin{align*}
r(k)=\Sigma_{\tilde{r}}^{-\frac{1}{2}}\left(\tilde{y}(k)-C\hat{x}(k)\right),
\end{align*}
where ${\Sigma_{\tilde{r}}=CPC^T+\Sigma_v}$ is the steady state covariance matrix of $\tilde{r}(k)=\tilde{y}(k)-C\hat{x}(k)$ under nominal conditions (no faults, no attacks), i.e., $\tilde{u}(k)=u(k)$ and $\tilde{y}(k)=y(k)$ for all $k$.
\begin{assumption}
\label{assum:DetectorConditions}
The detector dynamics \eqref{eq:DetectorModel} fulfill the subsequent three conditions:
    \begin{equation*}
		\begin{aligned}
			1)&\,\theta\big(x_D(k),r(k)\big)\ \text{and}\ d\big(x_D(k),r(k)\big)\ \text{are}\\
			&\,\text{continuous\ in}\ x_D(k)\ \text{and}\ r(k);\\
			2)&\,d\big(x_D(k),r(k)\big)\ \text{is\ coercive\footnotemark\ in}\ x_D(k)\ \text{and}\ r(k);\\
			3)&\,d(0,0) = 0\ \text{and}\ \theta(0,0) = 0.
		\end{aligned}
	\end{equation*}
\end{assumption}
\footnotetext{A function $g(x)$ is called coercive if $g(x)\rightarrow\infty$ as $\|x\|\rightarrow\infty$.}
If the predictions are accurate, i.e., $r(k)\approx 0$, the detector output should be small.
However, if the predictions are inaccurate, both the detector state and the detector output should increase.
Furthermore, the detector triggers an alarm whenever the detector output $y_D(k)$ exceeds the detection threshold $J_D>0$.

Since the detector input is a random variable under nominal conditions, that is, $r(k)\sim\mathcal{N}(0,I_{n_y})$ due to the Kalman filter, the detector output $y_D(k)$ is also a random variable. 
To avoid too frequent \emph{false alarms}, i.e., alarms triggered under nominal conditions when there is \emph{no attacker present}, the threshold $J_D$ should be chosen large enough. 
However, if it is chosen too large, the detector might not be able to detect anomalies (missed detection). 
Hence, there is a trade-off between false alarms and missed detections when choosing $J_D$. 
Urbina \emph{et al.}~\cite{Alvaro} further noted that there is also a trade-off between false alarms and the attack impact of a stealthy attack, which in our case corresponds to a missed detection.
For example, a larger $J_D$ reduces the frequency of false alarms but gives the attacker more space to remain stealthy while causing harm.

Since the amount of false alarms plays an important role in detector tuning, we denote by $\tau$ the \emph{mean time between false alarms}. 
The larger time between false alarms we want to achieve, the larger the detector threshold has to be such that the following is a reasonable assumption.

\begin{assumption}
    \label{assum:DetectorThresholdIncreasing}
    The detector threshold is a strictly increasing (possibly nonlinear) function of $\tau$, i.e., $J_D(\tau_a)<J_D(\tau_b)$ if, and only if, $\tau_a<\tau_b$.
\end{assumption}

Instead of considering the threshold, $J_D$, we will consider the mean time between false alarms, $\tau$, in the following, since there is a direct relation between $J_D$ and $\tau$. 
Further, the value of $\tau$ is more meaningful to the operator.
To circumvent the trade-off between false alarms and missed detections (and the impact of stealthy attacks), the defender could periodically randomize the choice of the mean time between false alarms $\tau$ such that in one period the threshold reduces the number of false alarms while in another it limits the impact of a potential stealthy attack.
That is, it chooses $\tau$ periodically from the fixed set $\lbrace \tau_1,\ldots,\tau_m\rbrace$, pre-determined by the defender, with probability distribution $p$, where $1\leq\tau_1<\tau_2<\ldots<\tau_m$, $p_{i}\in[0,1]$ is the probability of choosing $\tau_i$, and $\sum_{i=1}^mp_i=1$. 
We make the following assumption about the random choice of $\tau$.
\begin{assumption}
\label{assum:IIDSwitchingPeriod}
At the beginning of each period, $\tau$ is drawn from the probability distribution $p$, independent from previous realizations.
\end{assumption}

Our definition of a moving target defense is stated next.
\begin{definition}
    \label{def:MovingTargetDefense}
    A probability distribution $p\in\mathbb{R}^m$ over a fixed set of mean times between false alarms $\lbrace\tau_1,\ldots,\tau_m\rbrace$ is a \emph{moving target defense} if $p$ does not have singleton support, i.e., the probability of choosing $\tau_i$ fulfills $p_i\in[0,1)$ for all $i\in\lbrace 1,\cdots,m \rbrace$.
\end{definition}
\subsection{Attacker Model}
In this paper, our focus lies on sensor attacks.
\begin{assumption}
    The measurement signals are subject to an additive attack $y_a(k)$ chosen by the attacker, and the actuator signals are transmitted fault/attack-free, i.e., ${\tilde{y}(k)=y(k)+y_a(k)}$ and $\tilde{u}(k)=u(k)$.
\end{assumption}

Furthermore, we make the following assumption on the attacker's model knowledge.
\begin{assumption}
    \label{assum:AttackerKnowledge}
    The attacker knows the closed-loop system matrices, $A$, $B$, $C$, $L$, $K$, the noise statistics $\Sigma_w$, $\Sigma_v$, and the detector dynamics. 
    The attack starts at time $k=0$ and has a length of $N$ time steps.
    The length $N$ of the attack is such that the attacker is able to complete the attack  before the next threshold switch.
    The attacker knows both $x_D(0)$ and $\hat{x}(0)$, and has access to the measurements $y(k)$.
    Moreover, the attacker knows the function $J_D(\tau)$ and the set $\lbrace\tau_1,\ldots,\tau_m \rbrace$ but \emph{not} the exact value of $\tau$.
\end{assumption}

An attacker according to Assumption~\ref{assum:AttackerKnowledge} can launch an attack of the form (see~\cite{RuthsMultivariate})
\begin{align*}
    y_a(k)=-y(k)+C\hat{x}(k)+\Sigma_{\tilde{r}}^{\frac{1}{2}}a(k),
\end{align*}
which gives the attacker complete control over the detector input, i.e., $r(k)=a(k)$.
This attack is a \textit{closed-loop} attack since it uses the measurements $y(k)$, whereas $a(k)$ can be interpreted as the attacker's reference signal.
The attacker can define the set of attacks that do not trigger an alarm for a given $\tau$ as
\begin{align*}
   \resizebox{0.99\hsize}{!}{
        $\mathbb{A}(\tau):=
   \left\lbrace \lbrace a,x_D(0)\rbrace\ \Bigg| \  \begin{aligned}&x_D(k+1)=\theta\big(x_D(k),a(k)\big)\\ &y_D(k+1)=d\big(x_D(k),a(k)\big)\leq J_D(\tau)\\
    &k\in[0,N-1]\end{aligned}\right\rbrace$%
        },
\end{align*}
where $a=[a(0)^T,\cdots, a(N-1)^T]^T$ is the complete attack trajectory during the attack. 
Here, the set $\mathbb{A}(\tau)$ constrains the size of $y_a(k)$ as well by constraining $a(k)$.
\begin{remark}
Note that $\mathbb{A}(\tau_a)\subset\mathbb{A}(\tau_b)$ if $\tau_a<\tau_b$ due to Assumption~\ref{assum:DetectorThresholdIncreasing}. 
Furthermore, where appropriate we will use ${a\in\mathbb{A}(\tau)}$ instead of $\lbrace a,x_D(0)\rbrace\in\mathbb{A}(\tau)$ for the sake of readability.
\end{remark}

If the attacker manages to choose $a$ such that $\lbrace a,x_D(0)\rbrace\in\mathbb{A}(\tau)$ in the current period with a constant threshold, the attacker remains stealthy, which is the main constraint of the attacker as described below.

\begin{assumption}
\label{assum:AttackerGoal}
The attacker has one of $n_\phi$ different types, which determine the objective of the attacker.
An attacker of type $\phi$ wants to maximize its expected payoff characterized by $f_\phi(a)$.
Further, if the attack is detected, the attacker receives no payoff and, therefore, it wants to remain stealthy, i.e., $y_D(k+1)\leq J_D(\tau)$ for $k\in[ 0,N-1]$.
\end{assumption}

Next, we define the attacker's expected payoff for a given attacker type,
\begin{align}
\label{eq:AttackerPayoffGivenAttackType}
 \mathfrak{p}(\tau, a | \phi ):=\mathbf{1}_{\lbrace{\lbrace a,x_D(0)\rbrace\in \mathbb{A}(\tau)}\rbrace} f_{\phi}(a),
\end{align}
where we use the indicator function to take into account that the attacker will not get any payoff if it is detected.

\begin{assumption}
    \label{assum:PayoffFunctionOfAttacker}
    For a given attacker type $\phi$, the corresponding expected attacker payoff $f_\phi(a)$ is continuous and fulfills $\max_{a\in\mathbb{A}(\tau_a)}f_\phi(a)<\max_{a\in\mathbb{A}(\tau_b)}f_\phi(a)$ if $\tau_a<\tau_b$ except when $f_{\phi}(a)= 0$ for all $a$.
\end{assumption}

\begin{remark}
    If $f_{\phi}(a)$ is a continuous, convex function and $\mathbb{A}(\tau)$ is a closed convex set, then Assumption~\ref{assum:PayoffFunctionOfAttacker} is fulfilled, since then $\max_{a\in\mathbb{A}(\tau)}f_\phi(a)$ is equivalent to a concave minimization problem, whose optimizers are the extreme points of $\mathbb{A}(\tau)$ (see \cite{ConcaveMinimizationBook}).  
    If we use a vector norm-based stateless detector, such as the $\chi^2$ detector, $\mathbb{A}(\tau)$ is a closed convex set.
\end{remark}

\subsection{Defender model}
Next, we describe our defender model.
When choosing $\tau$ the defender needs to take into account the expected cost that is induced by the false alarms in the nominal case, but also the expected cost of an undetectable attack.
This leads to the following cost function for the defender assuming an attacker of type $\phi$,
\begin{align}
\label{eq:DefenderCostGivenAttackType}
\mathfrak{c}(\tau, a |\phi):=\frac{c_\text{F}}{\tau}+\mathfrak{p}(\tau, a | \phi ),
\end{align}
where $c_\text{F}>0$ is the cost factor for false alarms.
Note that while the attacker has one of $n_\phi$ possible types, the defender has only one type, but its cost function is influenced by the attacker type.
\begin{remark}
\label{rem:DummyAttackersMakesSense}
Since the defender's cost \eqref{eq:DefenderCostGivenAttackType} is always influenced by the attacker's payoff, it is reasonable to introduce an attacker type with zero payoff, i.e., $f_{\phi}(a)=0$ for all $a$. 
This means that the case of there not being an attacker present in the system is modeled as well in our moving target defense framework.
\end{remark}

Next, let us make the following assumption about the knowledge the attacker and defender have about each other.
\begin{assumption}
\label{assum:PlayerKnowledge}
The defender knows
\eqref{eq:AttackerPayoffGivenAttackType} for each possible $\phi$ and it has a prior $\pi_\phi$ of facing an attacker of type $\phi$, where $\pi_\phi\in[0,1]$ and $\sum_{\phi=1}^{n_\phi}\pi_{\phi}=1$.
The attacker, in addition to Assumption~\ref{assum:AttackerKnowledge}, knows the defender's cost function \eqref{eq:DefenderCostGivenAttackType}, its own type $\phi\in\lbrace 1,\ldots,n_\phi\rbrace$, and the defender's prior $\pi_\phi$ for each attacker type.
\end{assumption}

\subsection{Discussion of the system model}
\label{sec:DiscussionOfAssumptions}
In this section, we discuss the assumptions made during the setup of the model.

First, we discuss the assumption about iid choice of $\tau$ (Assumption~\ref{assum:IIDSwitchingPeriod}) and the attack length (Assumption~\ref{assum:AttackerKnowledge}).
Since the values of $\tau$ are realizations of iid random variables, the current value of $\tau$ is independent of its previous values, observing the system does hence not reveal information about $\tau$ beyond its distribution, which the attacker can deduce from its system knowledge.
The attacker will be able to estimate the distribution $p$ under the iid choice if it has access to previous values of $\tau$. This case is already taken into account in our MTD framework.
A change of threshold implies a reconfiguration of the system, which can be costly for the operator of a safety-critical large-scale infrastructure. 
Therefore, the operator does not switch the thresholds too frequently. 
Hence, it is not unreasonable to analyse the case where the attack is carried out during a fixed, but random, configuration.
Note that an approach to consider the cost of a finite amount of switches can be found in \cite{GameTheoryThresholds2016}.
Furthermore, for industrial processes, where a product is produced in batches, an iid choice of the threshold between different batches is also a reasonable assumption. 

The analysis of an attacker that experiences threshold switches during the attack is similar to the analysis we present in the subsequent sections, because we can determine the probability of first choosing $\tau_i$ and then $\tau_j$ due to the iid choice in Assumption~\ref{assum:IIDSwitchingPeriod}.
However, the notation would become more involved.
Therefore, these assumptions simplify the problem formulation so that it becomes mathematically more tractable.

Next, we justify the assumptions on the attacker's knowledge and goals.
According to \cite{SecrecySystemsShannon}, one should design the plant for the worst-case attacker knowledge, because, given enough time, an attacker may be able to obtain a perfect model of the plant, the controller, and the detector. 
For example, the plant and controller could be estimated through system identification techniques from the observed sensor data, while the detector model could be obtained from leaked documentation of the system.
Hence, the extensive knowledge of the attacker about the closed-loop system and detector according to Assumption~\ref{assum:AttackerKnowledge} is in line with \cite{SecrecySystemsShannon}.
In our previous work \cite{UmsonstConEngPract22}, we showed how the attacker can obtain the internal states of both the controller and the detector in an experimental setup. Hence, assuming that the attacker has knowledge of the controller and detector states is not unreasonable.
Further, the knowledge of $x_D(0)$ and $\hat{x}(0)$ can be interpreted as an opportunistic attacker choosing to attack at the best time instant, which we define to occur without loss of generality at $k=0$. 
In contrast, the choice of $\tau$ is not visible in the sensor data observed by the attacker.
In addition, the attacker knowledge in Assumption~\ref{assum:AttackerKnowledge} together with the assumption that the attacker will maximize its objective function $f_{\phi}(a)$ (Assumption~\ref{assum:AttackerGoal}) results in a worst-case scenario for the defender under the given assumption.

While both $x_D(0)$ and $\hat{x}(0)$ depend on the measurements and could, therefore, be estimated by the attacker, $\tau$ is chosen randomly from $\lbrace\tau_1,\ldots,\tau_m \rbrace$ (see Assumption~\ref{assum:IIDSwitchingPeriod}) such that the attacker cannot know the exact value of $\tau$.
Further, $\tau$ does not directly influence the system variables, such that the attacker is also not able to estimate $\tau$ from the measurements.

In Assumption~\ref{assum:AttackerGoal}, we introduce attacker types. A given attacker type, $\phi$, describes the target of the attacker through its objective, $f_{\phi}(a)$. 
Since the attacker's target is often not known to the defender, having different attacker types gives the defender the possibility to distinguish between different targets while using our sensor attack model, and also incorporate the case of no attacker being present (Remark~\ref{rem:DummyAttackersMakesSense}).

The assumption that the attacker will not get any payoff when detected (Assumption~\ref{assum:AttackerGoal}) is a strong assumption on both the attacker and defender, which is mostly beneficial for the defender.
However, if we consider critical infrastructures, such as the power grid, an operator has to mitigate the attack quickly when detected to prevent harm.
Furthermore, we can also imagine that the attacker has made a significant investment to obtain its system knowledge and infiltrate the system. 
Hence, the attacker wants to remain undetected in order to not risk losing its investment. 
This kind of attacker has similarities to an advanced persistent threat, which is an attacker with knowledge about the system and that targets specific parts of the system while remaining stealthy (see, for example, \cite{SecureSensorDesignAgainstAPTBasar}).
It is important to point out that due to the random choice of $\tau$ (Assumption~\ref{assum:IIDSwitchingPeriod}) it is more difficult for the attacker to remain stealthy but at the same time obtain a large payoff.

The defender will rarely know the intentions of the attacker. 
To obtain information about potential targets of the attack, the defender can conduct a risk assessment \cite{GuideForRiskAssessment} of the system.
By conducting a risk assessment, the defender determines the vulnerabilities in its system, the likelihood of an attacker exploiting a vulnerability, and the potential impact of a successful attack. 
A vulnerability could be an unsafe region in the system's state space, e.g., the overpressure region for a tank, such that the attacker's objective would be to bring the system into this unsafe region.
The different vulnerabilities can be interpreted as attacker types $\phi$, the prior $\pi_{\phi}$ as the likelihood of an attacker exploiting a vulnerability, and the impacts are reflected by the attacker's payoff $f_{\phi}(a)$, which directly influence the defender's cost \eqref{eq:DefenderCostGivenAttackType}.
Hence, the defender's knowledge about possible attack objectives, their prior and their impact, as assumed in Assumption~\ref{assum:PlayerKnowledge}, can be interpreted as the outcome of a risk assessment conducted by the defender.
Therefore, we can interpret the Bayesian moving target defense framework as a tool to enhance the security for the defender similar to the ARMOR framework deployed at LAX \cite{ARMORFramework}, which makes use of the outcomes of the risk assessment.

\section{Problem Formulation}
Now we formulate the problem of finding a moving target defense strategy as a game between the defender and the attacker, where the defender's goal is to choose $\tau$ to minimize the expected value of \eqref{eq:DefenderCostGivenAttackType} with respect to the prior of the attacker types while the attacker chooses $a$ to maximize \eqref{eq:AttackerPayoffGivenAttackType}.
Due to Assumption~\ref{assum:AttackerKnowledge}, we can focus on the game over one period with a constant threshold.
This focus on one only period can also be interpreted as a repeated game with memoryless players, which has been considered in \cite{MTDWithMemorylessPlayers}.

The game has both \emph{imperfect} and \emph{incomplete} information. 
The information is imperfect because neither player observes the action taken by the other player. 
The information is incomplete because the defender does not know which type of attacker it faces. 
The defender believes that with probability $\pi_\phi$ it will play the game with an attacker of type $\phi$.
The imperfect information lets us interpret the game as a game with simultaneous moves, while the incomplete information results in a Bayesian game framework.
Therefore, we define the moving target defense game $\mathcal{M}=\langle \mathcal{P}, \mathcal{A}, \mathcal{T} ,\Pi, \mathcal{U} \rangle$, where $\mathcal{P}=\lbrace \mathrm{Defender}, \mathrm{Attacker}\rbrace$ is the set of players, $\mathcal{A}=\lbrace\tau_1,\ldots,\tau_m\rbrace\times \mathbb{R}^{Nn_y+n_D}$ is the action set, $\mathcal{T}=\lbrace 1 \rbrace\times\lbrace 1,\ldots,n_{\phi}\rbrace$ is the set of player types, $\Pi=\lbrace 1\rbrace\times \lbrace\pi_1,\ldots,\pi_{n_{\phi}}\rbrace$ is the prior, and $\mathcal{U}=\left(\mathfrak{c}(\tau, a | \phi ),\ \mathfrak{p}(\tau, a | \phi )\right)$ contains the cost and payoff functions of each player.   
For the analysis, we also define the game $\mathcal{M}_\phi=\langle \mathcal{P}, \mathcal{A}, \mathcal{U} \rangle$, where the defender is certain about the attacker type it faces, that is, $\pi_\phi=1$ for some $\phi\in\lbrace 1,\ldots,n_{\phi}\rbrace$ in $\mathcal{M}$.

The Bayesian game framework together with the simultaneous choice of actions lead us to the \emph{Bayesian Nash equilibrium} as the solution concept.
To define the Bayesian Nash equilibrium we introduce the (possibly mixed) strategies of the defender and attacker.
Let $\Delta_p$ be the set of probability distributions over the defender's actions. Then $p\in\Delta_p$ is a discrete probability distribution, where the $i$th element, $p_i$, is the probability that the defender chooses $\tau_i$.
For a given attacker type $\phi$, $\Delta_q(\phi)$ is the set of probability distributions over the attacker's action set. 
Since the attacker, to obtain a non-zero payoff, chooses a trajectory $a$ from $\mathbb{A}(\tau)$, which is typically not a discrete set, $q_{\phi}\in\Delta_q(\phi)$ may represent a continuous probability distribution.
We call both $p$ and $q_{\phi}$ a \emph{mixed} strategy, if it does \emph{not} concentrate the whole probability on one action. 
Otherwise, we call it a \textit{pure} strategy.

Since both the attacker and defender might use mixed strategies, we investigate the average cost of the defender 
\begin{align}
    \label{eq:AverageDefenderCostWithMixedStrategies}
    \bar{\mathfrak{c}}_\phi(p,q_{\phi})=\int\sum_{i=1}^m p_i \mathfrak{c}(\tau_i,a|\phi)q_\phi(a)da
\end{align}
and the average payoff of the attacker
\begin{align}
    \label{eq:AverageAttackerPayoffWithMixedStrategies}
    \bar{\mathfrak{p}}_\phi(p,q_{\phi})=\int\sum_{i=1}^m p_i \mathfrak{p}(\tau_i,a|\phi)q_\phi(a)da
\end{align}
for a given attacker type $\phi$.
Hence, \eqref{eq:AverageDefenderCostWithMixedStrategies} and \eqref{eq:AverageAttackerPayoffWithMixedStrategies} represent the average cost and payoff of the players in the game $\mathcal{M}_\phi$.

    A \textit{mixed strategy Bayesian Nash equilibrium}, $p^*\in\Delta_p$ and $q^*_{\phi}\in\Delta_q(\phi)$, fulfills 
    \begin{equation}  
    \label{eq:BayesNashEquilibriumGeneral}
    \begin{aligned}
    		\sum_{\phi^{\prime}=1}^{n_{\phi}}\pi_{\phi^{\prime}}\bar{\mathfrak{c}}_{\phi^{\prime}}(p^*,q^*_{\phi^{\prime}})&\leq \sum_{\phi^{\prime}=1}^{n_{\phi}}\pi_{\phi^{\prime}}\bar{\mathfrak{c}}_{\phi^{\prime}}(p,q^*_{{\phi^{\prime}}}),\\
    		\bar{\mathfrak{p}}_\phi(p^*,q^*_{\phi})&\geq \bar{\mathfrak{p}}_\phi(p^*,q_{\phi})
    \end{aligned}
    \end{equation}
     for all $p\in\Delta_p$, $q_{\phi}\in\Delta_q(\phi)$, and $\phi\in\lbrace1,\cdots,n_{\phi}\rbrace$.

In the Bayesian Nash equilibrium, a change from $p^*$ to another $p\in\Delta_p$ does not lead to a decrease in the cost for the defender, and, similarly, a change from $q_{\phi}^*$ to another $q_{\phi}\in\Delta_q(\phi)$ does not lead to an increase in payoff for an attacker of type $\phi$.
Hence, neither the defender nor the attacker want to deviate from their equilibrium strategies.
Here, the defender needs to consider all possible attacker types, which results in averaging of the costs of each game $\mathcal{M}_\phi$ over the prior, while the attacker needs to have an equilibrium strategy for each type.
This is because the attacker knows its own type, which the defender does not know, while the defender has only one type, which is known to both the defender and attacker.

Equipped with the definition of both the MTD and the Bayesian Nash equilibrium, we formulate the two problems we investigate in the remainder of this paper.
\begin{problem}
    \label{prob:ExistenceOfBayesNashEquilibrium}
    Characterize when a Bayesian Nash equilibrium \eqref{eq:BayesNashEquilibriumGeneral} representing a MTD (Definition~\ref{def:MovingTargetDefense}) exists.
\end{problem}
\begin{problem}
    \label{prob:DetermineBayesNashEquilibrium}
    If a Bayesian Nash equilibrium representing a MTD exists, compute an equilibrium strategy $p^*$.
\end{problem}
\section{Matrix Game Formulation}
\label{sec:MatrixGameFormulation}

Recall that the defender plays against one of $n_\phi$ adversaries, but it does not know which adversary it is facing.
Furthermore, while the defender has a finite set of actions, i.e., $\lbrace\tau_1,\cdots,\tau_m\rbrace$, the attacker's action set, $\mathbb{R}^{Nn_y+n_D}$, is a continuum. 
This makes finding Bayesian Nash equilibrium strategies challenging.
In this section, we will show that each game $\mathcal{M}_\phi$ can be reformulated into a strategically equivalent game $\widetilde{\mathcal{M}}_\phi$, where the attacker's action set is finite too.

We begin by recalling that for a given $\tau$ the attacker will only receive a non-zero payoff if $\lbrace a, x_D(0)\rbrace\in\mathbb{A}(\tau)$. 
Hence, for a given $\tau$ the attacker will always choose its attack trajectory such that $\lbrace a, x_D(0)\rbrace\in\mathbb{A}(\tau)$. 
Due to the discrete set of actions for the defender, we can separate the continuous action space of the attacker into $m+1$ sets in the game $\mathcal{M}_\phi$ as shown in Table~\ref{tab:MatrixGameDetectorTuningPre}. 
The set $\mathbb{A}(\tau_i)\setminus\mathbb{A}(\tau_{i-1})$ contains all attack trajectories that are stealthy for $\tau_i$ excluding the ones that are stealthy for $\tau_{i-1}$. 
Hence, if $\lbrace a, x_D(0)\rbrace\in\mathbb{A}(\tau_i)\setminus\mathbb{A}(\tau_{i-1})$ then the attack will be detected if the defender chooses $\tau_{i-1}$, but not if it chooses $\tau_i$.
We can remove the last column from Table~\ref{tab:MatrixGameDetectorTuningPre}, because $a\not\in\mathbb{A}(\tau_m)$ results in zero payoff for the attacker.

\begin{table*}
\centering
\caption{The game $\mathcal{M}_\phi$ with disjoint sets for the attacker's action}
\label{tab:MatrixGameDetectorTuningPre}
\begin{tabular}{|c|c|c|c|c|c|c|}
\hline 
 & $a_1\in\mathbb{A}(\tau_1)$ & $a_2\in\mathbb{A}(\tau_2)\setminus\mathbb{A}(\tau_1)$ & $\cdots$ & $a_m\in\mathbb{A}(\tau_m)\setminus\mathbb{A}(\tau_{m-1})$ & $a\not\in\mathbb{A}(\tau_m)$ \\ 
\hline 
$\tau_1$ & $\frac{c_\text{F}}{\tau_1}+f_{\phi}(a_1),\ f_{\phi}(a_1)$ & $\frac{c_\text{F}}{\tau_1},\ 0$ & $\cdots$ & $\frac{c_\text{F}}{\tau_1},\ 0$ & $\frac{c_\text{F}}{\tau_1},\ 0$ \\ 
\hline 
$\tau_2$ & $\frac{c_\text{F}}{\tau_2}+f_{\phi}(a_1),\ f_{\phi}(a_1)$ & $\frac{c_\text{F}}{\tau_2}+f_{\phi}(a_2),\ f_{\phi}(a_2)$ & $\cdots$ & $\frac{c_\text{F}}{\tau_2},\ 0$ & $\frac{c_\text{F}}{\tau_2},\ 0$ \\ 
\hline 
$\vdots$ & $\vdots$ & $\vdots$ & $\ddots$ & $\vdots$ & $\vdots$ \\ 
\hline 
$\tau_m$ & $\frac{c_\text{F}}{\tau_m}+f_{\phi}(a_1),\ f_{\phi}(a_1)$ & $\frac{c_\text{F}}{\tau_m}+f_{\phi}(a_2),\ f_{\phi}(a_2)$  & $\cdots$ & $\frac{c_\text{F}}{\tau_m}+f_{\phi}(a_m),\ f_{\phi}(a_m)$ & $\frac{c_\text{F}}{\tau_m},\ 0$ \\ 
\hline 
\end{tabular} 
\end{table*}

We define the maximum payoff for a given $\tau_i$ and a given attacker type $\phi$ as
\begin{align*}
 \mathcal{I}_i^\phi:=\max_{a,x_D(0)} \mathbf{1}_{\lbrace \lbrace a, x_D(0)\rbrace\in\mathbb{A}(\tau_i)\rbrace}f_\phi(a)=\max_{\lbrace a,x_D(0)\rbrace\in\mathbb{A}(\tau_i)} f_\phi(a).
\end{align*}
Note that we also optimize over $x_D(0)$, which the attacker has normally no influence over. 
We do that to obtain the maximum possible payoff an attacker could achieve, which goes along with the scenario of the worst-case attacker and the interpretation that the attacker waits for the optimal time to attack.
We can show the following for the maximum payoff.
\begin{lemma}
    \label{lem:ExistenceOfImpact}
    Consider an attacker of type $\phi$. For a given $\tau_i$, the maximum payoff exists and
    $\mathcal{I}_i^\phi:=\max_{\lbrace a,x_D(0)\rbrace\in\mathbb{A}(\tau_i)} f_\phi(a)
    =\max_{\lbrace a,x_D(0)\rbrace\in\mathbb{A}(\tau_i)\setminus\mathbb{A}(\tau_{i-1})}f_\phi(a).$
\end{lemma}

\begin{proof}
    We begin by proving the first part of the lemma.
    The first and second condition in Assumption~\ref{assum:DetectorConditions} guarantee that $\mathbb{A}(\tau_i)$ is a compact set for a given $\tau_i$ (see Theorem~7.1 in \cite{UmsonstLic}). 
    Hence, by the extreme value theorem, we know that $\mathcal{I}^\phi_i$ always exists for a given $\tau_i$ and $\phi$.
    The second part of the lemma readily follows from Assumption~\ref{assum:PayoffFunctionOfAttacker}.
\end{proof}
With the maximum payoff for a given $\tau_i$, we can formulate a finite matrix game $\widetilde{\mathcal{M}}_{\phi}=\langle \mathcal{P}, \widetilde{\mathcal{A}}_{\phi}, \widetilde{\mathcal{U}}_{\phi}\rangle$, where $\mathcal{P}$ is defined as in $\mathcal{M}$, $\widetilde{\mathcal{A}}_{\phi}=\lbrace\tau_1,\ldots,\tau_m\rbrace\times\lbrace\mathcal{I}_1^{\phi},\ldots,\mathcal{I}_m^{\phi}\rbrace$, and $\widetilde{\mathcal{U}}_{\phi}=(\frac{c_\text{F}}{\tau_i}+\mathbf{1}_{\lbrace j\leq i\rbrace}\mathcal{I}^\phi_j,\mathbf{1}_{\lbrace j\leq i\rbrace}\mathcal{I}^\phi_j)$, where $i,j\in\lbrace 1,\ldots,m\rbrace$.
Since both the attacker and the defender have finite actions sets in $\widetilde{M}_{\phi}$, we formulate $\widetilde{M}_{\phi}$ as the matrix game shown in Table~\ref{tab:MatrixGameDetectorTuning}.
Further, we define the $m\times m$ matrix $\Omega(\phi)$ as the defender's cost matrix with elements $\Omega_{i,j}(\phi)=\frac{c_\text{F}}{\tau_i}+\mathbf{1}_{\lbrace j\leq i\rbrace}\mathcal{I}^\phi_j$, and $\Upsilon(\phi)$ as the $m\times m$ matrix that has the attacker's payoff matrix with elements
$\Upsilon_{i,j}(\phi)=\mathbf{1}_{\lbrace j\leq i\rbrace}\mathcal{I}^\phi_j$.
\begin{proposition}
    \label{prop:EquivalentMatrixGame}
    The finite game $\widetilde{\mathcal{M}}_{\phi}$ in Table~\ref{tab:MatrixGameDetectorTuning} is strategically equivalent to the game $\mathcal{M}_\phi$ in Table~\ref{tab:MatrixGameDetectorTuningPre}.
\end{proposition}
\begin{proof}
    Since the attacker's objective is to maximize its payoff \eqref{eq:AttackerPayoffGivenAttackType}, it always chooses the trajectory that maximizes its payoff. 
    From Lemma~\ref{lem:ExistenceOfImpact} we know there exists a maximum payoff trajectory for each of the columns in Table~\ref{tab:MatrixGameDetectorTuningPre}.
    Hence, choosing the maximum payoff is strategically equivalent to choosing an attack trajectory that yields the maximum payoff.
\end{proof}

\begin{table}[h]
\centering
\caption{Finite $m \times m$ matrix game $\widetilde{\mathcal{M}}_{\phi}$ for threshold switching}
\label{tab:MatrixGameDetectorTuning}
\begin{tabular}{|c|c|c|c|c|c|c|}
\hline 
 & $\mathcal{I}_1^\phi$ &$\mathcal{I}_2^\phi$ & $\cdots$ & $\mathcal{I}_m^\phi$  \\ 
\hline 
$\tau_1$ & $\frac{c_\text{F}}{\tau_1}+\mathcal{I}_1^\phi,\ \mathcal{I}_1^\phi$ & $\frac{c_\text{F}}{\tau_1},\ 0$ & $\cdots$ & $\frac{c_\text{F}}{\tau_1},\ 0$ \\ 
\hline 
$\tau_2$ & $\frac{c_\text{F}}{\tau_2}+\mathcal{I}_1^\phi,\ \mathcal{I}_1^\phi$ & $\frac{c_\text{F}}{\tau_2}+\mathcal{I}_2^\phi,\ \mathcal{I}_2^\phi$ & $\cdots$ & $\frac{c_\text{F}}{\tau_2},\ 0$ \\ 
\hline 
$\vdots$ & $\vdots$ & $\vdots$ & $\ddots$ & $\vdots$ \\ 
\hline 
$\tau_m$ & $\frac{c_\text{F}}{\tau_m}+\mathcal{I}_1^\phi,\ \mathcal{I}_1^\phi$ & $\frac{c_\text{F}}{\tau_m}+\mathcal{I}_2^\phi,\ \mathcal{I}_2^\phi$  & $\cdots$ & $\frac{c_\text{F}}{\tau_m}+\mathcal{I}_m^\phi,\ \mathcal{I}_m^\phi$ \\ 
\hline 
\end{tabular} 
\end{table}

By using the equivalent game in Table~\ref{tab:MatrixGameDetectorTuning}, we can simplify the average cost functions, \eqref{eq:AverageDefenderCostWithMixedStrategies} and \eqref{eq:AverageAttackerPayoffWithMixedStrategies} of the game $\mathcal{M}_\phi$, used in the Bayesian Nash equilibrium \eqref{eq:BayesNashEquilibriumGeneral} to bilinear functions of $p$ and $q_{\phi}$, which helps us to solve both Problem~\ref{prob:ExistenceOfBayesNashEquilibrium} and Problem~\ref{prob:DetermineBayesNashEquilibrium}.
\begin{corollary}
In the strategically equivalent finite game $\widetilde{\mathcal{M}}_{\phi}$, the average cost of the defender is given by
\begin{align*}
    \bar{\mathfrak{c}}_{\phi}(p,q_\phi)=p^T\Omega(\phi)q_{\phi}
\end{align*}
and the average payoff of the attacker is given by
\begin{align*}
    \bar{\mathfrak{p}}_{\phi}(p,q_\phi)=p^T\Upsilon(\phi)q_{\phi}
\end{align*}
for each attacker type, where the $i$th element, $q_{\phi,i}$, of $q_{\phi}$ is the probability of choosing an attack trajectory that leads to the maximum payoff $\mathcal{I}_{i}^{\phi}$.
\end{corollary}
\begin{proof}
    Since the attacker has a finite set of actions in $\widetilde{\mathcal{M}}_{\phi}$, its mixed strategy $q_{\phi}$ is a discrete probability distribution. 
    This leads directly to bilinear functions of $p$ and $q_{\phi}$ for the average cost and payoff, respectively.
\end{proof}
\section{Bayesian Nash Equilibrium-based MTD}
\label{sec:BayesianNashEquilibrium}
In the previous section, we showed that for any particular $\phi$ the corresponding game $\mathcal{M}_\phi$ is strategically equivalent to a finite matrix game $\widetilde{\mathcal{M}}_{\phi}$.
This means that the Bayesian game $\mathcal{M}$ is strategically equivalent to a finite Bayesian game, denoted by $\widetilde{\mathcal{M}}$ and its equilibria can be found by formulating an induced matrix game \cite{EssentialsOfGameTheory}, obtained by combining the games $\widetilde{\mathcal{M}}_{\phi}$ with respect to the prior. In what follows, we first illustrate the procedure and we then use the induced game to give
a necessary and sufficient condition for the existence of a Bayesian Nash equilibrium that is a moving target defense according to Definition~\ref{def:MovingTargetDefense}. 

\subsection{An illustrative example}
We start with an illustrative example, where each player has two actions to choose from. 
The attacker is assumed to have type~1 with probability $\pi_1$ and type~2 with probability $\pi_2=1-\pi_1$.
Hence, the finite game $\widetilde{\mathcal{M}}_\phi$ corresponding to attacker type $\phi$ is as shown in Table~\ref{tab:MatrixGameExample2by2}, where $\phi\in\lbrace 1, 2\rbrace$.

\begin{table}[ht]
\centering
\caption{$2 \times 2$ matrix game example}
\label{tab:MatrixGameExample2by2}
\begin{tabular}{|c|c|c|c|c|c|c|}
\hline 
 & $\mathcal{I}_1^\phi$ &$\mathcal{I}_2^\phi$  \\ 
\hline 
$\tau_1$ & $\frac{c_\text{F}}{\tau_1}+\mathcal{I}_1^\phi,\ \mathcal{I}_1^\phi$ & $\frac{c_\text{F}}{\tau_1},\ 0$ \\ 
\hline 
$\tau_2$ & $\frac{c_\text{F}}{\tau_2}+\mathcal{I}_1^\phi,\ \mathcal{I}_1^\phi$ & $\frac{c_\text{F}}{\tau_2}+\mathcal{I}_2^\phi,\ \mathcal{I}_2^\phi$ \\
\hline
\end{tabular} 
\end{table}

To find the Bayesian Nash equilibrium, we can formulate an induced matrix game (see \cite{EssentialsOfGameTheory}) and find the Nash equilibria of that induced matrix game, which correspond to the Bayesian Nash equilibria of the original game $\mathcal{M}$.
In the induced game, we combine the matrix games $\widetilde{\mathcal{M}}_1$ and $\widetilde{\mathcal{M}}_2$ into one game. 
The actions of the defender in the induced game are the same as in the games $\widetilde{\mathcal{M}}_1$ and $\widetilde{\mathcal{M}}_2$, that is, it can choose $\tau_1$ or $\tau_2$ as its action. 
The attacker, however, has the actions $\mathcal{I}_{i_1}^1\mathcal{I}_{i_2}^2$, where $i_1$ and $i_{2}$ are in $\lbrace 1,2\rbrace$. 
Hence, the attacker in the induced game is a combination of the attackers in the games $\widetilde{\mathcal{M}}_1$ and $\widetilde{\mathcal{M}}_2$ and its payoff is the expected value over the attacker types given the defender's prior $[\pi_1,\,\pi_2]$. 
The induced game is illustrated in Table~\ref{tab:InducedMatrixGame2by2}.
If the attacker chooses $\mathcal{I}_{i_1}^1\mathcal{I}_{i_2}^2$ in the induced game, then in $\widetilde{\mathcal{M}}_1$ the action of the attacker is its $i_1$th action, i.e., $\mathcal{I}_{i_1}^1$, and in $\widetilde{\mathcal{M}}_2$ the action of the attacker is its $i_2$th action, i.e., $\mathcal{I}_{i_2}^2$.
\begin{table*}[!ht]
\centering
\caption{Induced matrix game  for the $2 \times 2$ game example}
\label{tab:InducedMatrixGame2by2}
\begin{tabular}{|c|c|c|c|c|c|c|}
\hline 
 & $\mathcal{I}_1^1\mathcal{I}_1^2$ &$\mathcal{I}_1^1\mathcal{I}_2^2$ & $\mathcal{I}_2^1\mathcal{I}_1^2$ & $\mathcal{I}_2^1\mathcal{I}_2^2$ \\ 
\hline 
$\tau_1$ & $\frac{c_\text{F}}{\tau_1}+\pi_1\mathcal{I}_1^1+\pi_2\mathcal{I}_1^2,\pi_1\mathcal{I}_1^1+\pi_2\mathcal{I}_1^2$ & $\frac{c_\text{F}}{\tau_1}+\pi_1\mathcal{I}_1^1,\pi_1\mathcal{I}_1^1$ & $\frac{c_\text{F}}{\tau_1}+\pi_2\mathcal{I}_1^2,\pi_2\mathcal{I}_1^2$ & $\frac{c_\text{F}}{\tau_1},0$ \\ 
\hline 
$\tau_2$ & $\frac{c_\text{F}}{\tau_2}+\pi_1\mathcal{I}_1^1+\pi_2\mathcal{I}_1^2,\pi_1\mathcal{I}_1^1+\pi_2\mathcal{I}_1^2$ & $\frac{c_\text{F}}{\tau_2}+\pi_1\mathcal{I}_1^1+\pi_2\mathcal{I}_2^2,\pi_1\mathcal{I}_1^1+\pi_2\mathcal{I}_2^2$ & $\frac{c_\text{F}}{\tau_2}+\pi_1\mathcal{I}_2^1+\pi_2\mathcal{I}_1^2,\pi_1\mathcal{I}_2^1+\pi_2\mathcal{I}_1^2$ & $\frac{c_\text{F}}{\tau_2}+\pi_1\mathcal{I}_2^1+\pi_2\mathcal{I}_2^2,\pi_1\mathcal{I}_2^1+\pi_2\mathcal{I}_2^2$ \\ 
\hline 
\end{tabular} 
\end{table*}
From Table~\ref{tab:InducedMatrixGame2by2}, we observe that the defender prefers $\tau_2$ over $\tau_1$ if the following conditions hold
\begin{align*}
    \frac{c_\text{F}}{\tau_1}+\pi_1\mathcal{I}_1^1+\pi_2\mathcal{I}_1^2&>\frac{c_\text{F}}{\tau_2}+\pi_1\mathcal{I}_1^1+\pi_2\mathcal{I}_1^2,\\
     \frac{c_\text{F}}{\tau_1}+\pi_1\mathcal{I}_1^1&>\frac{c_\text{F}}{\tau_2}+\pi_1\mathcal{I}_1^1+\pi_2\mathcal{I}_2^2,\\
     \frac{c_\text{F}}{\tau_1}+\pi_2\mathcal{I}_1^2&>\frac{c_\text{F}}{\tau_2}+\pi_1\mathcal{I}_2^1+\pi_2\mathcal{I}_1^2,\\
    \frac{c_\text{F}}{\tau_1}&>\frac{c_\text{F}}{\tau_2}+\pi_1\mathcal{I}_2^1+\pi_2\mathcal{I}_2^2,
\end{align*}
which are equivalent to
\begin{align*}
    &\frac{c_\text{F}}{\tau_1}>\frac{c_\text{F}}{\tau_2},\\
    & \frac{c_\text{F}}{\tau_1}>\frac{c_\text{F}}{\tau_2}+\pi_2\mathcal{I}_2^2,\\
    & \frac{c_\text{F}}{\tau_1}>\frac{c_\text{F}}{\tau_2}+\pi_1\mathcal{I}_2^1,\\
    & \frac{c_\text{F}}{\tau_1}>\frac{c_\text{F}}{\tau_2}+\pi_1\mathcal{I}_2^1+\pi_2\mathcal{I}_2^2.
\end{align*}
Note that the first inequality always holds, while the second and third inequalities hold if the last inequality holds.

Hence, we see that the defender prefers to play $\tau_2$ over playing $\tau_1$ if  $\frac{c_\text{F}}{\tau_1}>\frac{c_\text{F}}{\tau_2}+\pi_1\mathcal{I}_2^1+\pi_2\mathcal{I}_2^2$.
In this case, the defender will play $\tau_2$ independent of the attacker's action, such that the attacker will always play the action that maximizes its payoff, i.e., $\mathcal{I}_2^1\mathcal{I}_2^2$. 
Therefore, there exists only a pure Bayesian Nash equilibrium strategy, which is \emph{not} a moving target defense.
For this simple example, we determined a sufficient condition for when an MTD does not exist.
However, this is a simple example where the induced matrix game has a managable size and we can calculate the Bayesian Nash equilibrium by hand.
Assume now that the defender has $m>2$ actions, while the attacker has $m$ actions and $n_\phi>1$ types.
Then the induced matrix game is an $m \times m^{n_{\phi}}$ matrix game, whose size becomes unmanagable as either $m$, $n_{\phi}$, or both, grow.

\subsection{Best responses and strictly dominated actions}
In the induced matrix game, the actions of the attacker are $\mathcal{I}_{i_1}^1\mathcal{I}_{i_2}^2\cdots\mathcal{I}_{i_{n_\phi}}^{n_\phi}$, where $i_\phi\in\lbrace1,\ldots, m\rbrace$ and $\phi\in\lbrace1,\ldots,n_{\phi}\rbrace$, while the defender chooses $\tau_l$. This leads to the attacker payoff
\begin{align*}
    \mathfrak{p}^{\mathrm{ind}}(\tau_l,\mathcal{I}_{i_1}^1\mathcal{I}_{i_2}^2\cdots\mathcal{I}_{i_{n_\phi}}^{n_\phi})=\sum_{j=1}^{n_\phi}\mathbf{1}_{\lbrace i_j\leq l\rbrace}\pi_j\mathcal{I}_{i_j}^j
\end{align*}
and the defender cost
\begin{align*}
    \mathfrak{c}^{\mathrm{ind}}(\tau_l,\mathcal{I}_{i_1}^1\mathcal{I}_{i_2}^2\cdots\mathcal{I}_{i_{n_\phi}}^{n_\phi})=\frac{c_\text{F}}{\tau_l}+\sum_{j=1}^{n_\phi}\mathbf{1}_{\lbrace i_j\leq l\rbrace}\pi_j\mathcal{I}_{i_j}^j
\end{align*}
in the induced matrix game, which we can use to characterize the best responses of the players.
\begin{lemma}
\label{lem:BestResponses}
    The best response of the attacker to a given action $\tau_l$ is
    \begin{align}
        \label{eq:AttackerBestResponseInducedMatrixGame}
        b_\text{A}(\tau_l)=\lbrace\mathcal{I}_{l}^1\mathcal{I}_{l}^2\cdots\mathcal{I}_{l}^{n_\phi}\rbrace,
    \end{align}
    and the best response of the defender to a given action $\mathcal{I}_{i_1}^1\mathcal{I}_{i_2}^2\cdots\mathcal{I}_{i_{n_\phi}}^{n_\phi}$ is
    \begin{small}
    \begin{align}
        \label{eq:DefenderBestResponseInducedMatrixGame}
        \resizebox{0.99\hsize}{!}{
        $b_\text{D}(\mathcal{I}_{i_1}^1\mathcal{I}_{i_2}^2\cdots\mathcal{I}_{i_{n_\phi}}^{n_\phi})=\big\lbrace \tau_l| l\in\arg\min_{l\in\lbrace 1,\cdots,m\rbrace}\frac{c_\text{F}}{\tau_l}+\sum_{j=1}^{n_\phi}\mathbf{1}_{\lbrace i_j\leq l\rbrace}\pi_j\mathcal{I}_{i_j}^j\big\rbrace$%
        }.
    \end{align}
    \end{small}
\end{lemma}
\begin{proof}
    We start by investigating the best response of the attacker. 
    For a given $\tau_l$, the payoff $\mathcal{I}_{i_j}^j$ of type $j$ with $i_j>l$ is zero since it is detected. 
    Hence, $i_j\leq l$ needs to be fulfilled for each attacker type if it wants to get a payoff. 
    Recall that $\mathcal{I}_{i}^j<\mathcal{I}_{\eta}^j$ for all $i<\eta$. 
    Hence, to obtain the maximum payoff for a given $\tau_l$, the attacker has a unique best response given by \eqref{eq:AttackerBestResponseInducedMatrixGame} in the induced matrix game.
    
    For a given attack action, the defender's best response is to choose $\tau_l$ to minimize its cost, which results in the set of best responses given in \eqref{eq:DefenderBestResponseInducedMatrixGame}.
\end{proof}
While Lemma~\ref{lem:BestResponses} provides the unique best response of the attacker, the defender might have several best responses.
For example, it could be best to choose $\tau_m$ for a given attacker action to minimize the cost for false alarms.
By choosing a smaller $\tau$, even though it increases the false alarm cost, it makes more attacks detectable, which in turn decreases the attack cost.
Hence, the best response depends on many factors.

Now that we looked at best responses, we will investigate when actions are strictly dominated in the induced matrix game. 
For the defender, an action $\tau_l$ strictly dominates $\tau_{\eta}$ if the cost for $\tau_l$ is strictly lower than the cost for $\tau_{\eta}$ for all possible actions of the attacker. 
Strict dominance of one attacker action over another can be defined similarly.
With the best responses and the strictly dominated actions, we are then equipped to prove the existence of moving target defenses according to Definition~\ref{def:MovingTargetDefense}.
Recall that by eliminating strictly dominated actions, we do not change the set of the Nash equilibria of the induced game and, therefore, neither the Bayesian Nash equilibria of the original game $\mathcal{M}$.

\begin{lemma}
    \label{lem:DominatedStrategiesInducedMatrixGame}
    Assume that there exists $l>{\eta}$ such that
    \begin{align}
        \label{eq:DominatedStrategiesInducedMatrixGame}
        \frac{c_\text{F}}{\tau_{\eta}}>\frac{c_\text{F}}{\tau_l}+\sum_{j=1}^{n_{\phi}}\pi_j\mathcal{I}_{l}^j,
    \end{align}
    then we can eliminate the defender actions $\tau_1,\ldots,\tau_{\eta}$ and the attacker actions for which  $\mathcal{I}_{i_1}^1\mathcal{I}_{i_2}^2\cdots\mathcal{I}_{i_{n_\phi}}^{n_\phi}\neq \mathcal{I}_{{\eta}+1}^1\mathcal{I}_{{\eta}+1}^2\cdots\mathcal{I}_{{\eta}+1}^{n_\phi}$ holds for $i_j\in\lbrace1,\cdots,{\eta}+1\rbrace$ ,
    and obtain a reduced $(m-{\eta})\times (m^{n_{\phi}}-({\eta}+1)^{n_{\phi}}+1)$ induced matrix game.
\end{lemma}
\begin{proof}
We start with the strict dominance of the rows. 
First note that $\tau_l$ strictly dominates $\tau_{\eta}$ if
\begin{align}
    \label{eq:CompleteConditionForStrictDominance}
    \frac{c_\text{F}}{\tau_{\eta}}+\sum_{j=1}^{n_\phi}\mathbf{1}_{\lbrace i_j\leq {\eta}\rbrace}\pi_j\mathcal{I}_{i_j}^j>\frac{c_\text{F}}{\tau_l}+\sum_{j=1}^{n_\phi}\mathbf{1}_{\lbrace i_j\leq l\rbrace}\pi_j\mathcal{I}_{i_j}^j
\end{align}
holds for all possible attacker actions $\mathcal{I}_{i_1}^1\mathcal{I}_{i_2}^2\cdots\mathcal{I}_{i_{n_\phi}}^{n_\phi}$.
We can split this condition into three cases:
\begin{enumerate}
    \item All attacker types use attacks that are stealthy for $\tau_{\eta}$, i.e., actions $\mathcal{I}_{i_1}^1\mathcal{I}_{i_2}^2\cdots\mathcal{I}_{i_{n_\phi}}^{n_\phi}$ with $\sum_{j=1}^{n_\phi}\mathbf{1}_{\lbrace i_j\leq {\eta}\rbrace}=n_{\phi}$.
    
    Here, we see that the terms related to the attacker payoff on both sides of the inequality in \eqref{eq:CompleteConditionForStrictDominance} are the same, such that \eqref{eq:CompleteConditionForStrictDominance} simplifies to $\frac{c_\text{F}}{\tau_{\eta}}<\frac{c_\text{F}}{\tau_l}$. 
    Hence, ${\eta}<l$ is necessary for the strict dominance of $\tau_l$ over $\tau_{\eta}$.
    
    \item All attacker types use attacks that will be detected for $\tau_l$, i.e., actions $\mathcal{I}_{i_1}^1\mathcal{I}_{i_2}^2\cdots\mathcal{I}_{i_{n_\phi}}^{n_\phi}$ with $\sum_{j=1}^{n_\phi}\mathbf{1}_{\lbrace i_j>l\rbrace }=n_{\phi}$.
    
    Since the attacks are detected by $\tau_l$, they will also be detected by $\tau_{\eta}$ such that the terms related to the attacker payoff on both sides of the inequality in \eqref{eq:CompleteConditionForStrictDominance} disappear. 
    Therefore, \eqref{eq:CompleteConditionForStrictDominance} simplifies again to $\frac{c_\text{F}}{\tau_{\eta}}<\frac{c_\text{F}}{\tau_l}$.
    
    \item There exists at least one attacker type that uses a strategy that is stealthy for $\tau_l$ but not for $\tau_{\eta}$, i.e., actions $\mathcal{I}_{i_1}^1\mathcal{I}_{i_2}^2\cdots\mathcal{I}_{i_{n_\phi}}^{n_\phi}$ with $\sum_{j=1}^{n_\phi}\mathbf{1}_{\lbrace {\eta}+1\leq i_j\leq l\rbrace}>0$.
    
    Subtracting the terms related to the attacker payoff on the left side of \eqref{eq:CompleteConditionForStrictDominance} from the inequality itself leads to
    \begin{align}
        \label{eq:ThirdCaseOfDominanceCondition}
        \frac{c_\text{F}}{\tau_{\eta}}>\frac{c_\text{F}}{\tau_l}+\sum_{j=1}^{n_\phi}\mathbf{1}_{\lbrace {\eta}+1\leq i_j\leq l\rbrace}\pi_j\mathcal{I}_{i_j}^j.
    \end{align}
    \end{enumerate}

The first two cases show that we need $\tau_l>\tau_{\eta}$ or equivalently $l>{\eta}$ for $\tau_l$ to strictly dominate $\tau_{\eta}$. For the third case, 
since 
\begin{align*}
\frac{c_\text{F}}{\tau_l}+\sum_{j=1}^{n_\phi}\pi_j\mathcal{I}_{l}^j\geq \frac{c_\text{F}}{\tau_l}+\sum_{j=1}^{n_\phi}\mathbf{1}_{\lbrace {\eta}+1\leq i_j\leq l\rbrace }\pi_j\mathcal{I}_{i_j}^j
\end{align*}
is always correct, \eqref{eq:ThirdCaseOfDominanceCondition} holds if \eqref{eq:DominatedStrategiesInducedMatrixGame} holds.

Further, since the following inequalities
\begin{align*}
    \frac{c_\text{F}}{\tau_\nu}+\sum_{j=1}^{n_\phi}\mathbf{1}_{\lbrace i_j\leq \nu\rbrace }\pi_j\mathcal{I}_{i_j}^j>\frac{c_\text{F}}{\tau_\nu}>\frac{c_\text{F}}{\tau_{\eta}}>\frac{c_\text{F}}{\tau_l}+\sum_{j=1}^{n_\phi}\pi_j\mathcal{I}_{l}^j
\end{align*}
hold for all $i_j\in\lbrace1,\ldots, m\rbrace$ and $j\in\lbrace1,\ldots,n_\phi\rbrace$,
we see that if \eqref{eq:DominatedStrategiesInducedMatrixGame} holds $\tau_l$ does not only strictly dominate $\tau_{\eta}$, but all $\tau_{\nu}$ with $\nu\in\lbrace1,\ldots, {\eta}\rbrace$.

Therefore, if \eqref{eq:DominatedStrategiesInducedMatrixGame} holds we can remove the first ${\eta}$ rows of the induced matrix game.
With the first ${\eta}$ rows removed it follows that $\mathcal{I}_{{\eta}+1}^1\mathcal{I}_{{\eta}+1}^2\cdots\mathcal{I}_{{\eta}+1}^{n_\phi}$ strictly dominates all actions where the attacker of type $j$ chooses $i_j\in\lbrace1,\cdots,{\eta}+1\rbrace$ with $j\in\lbrace1,\ldots,n_{\phi}\rbrace$ such that ${\sum_{j=1}^{n_{\phi}}\mathbf{1}_{\lbrace i_{j}={\eta}+1\rbrace}\neq n_\phi}$ holds, i.e., all actions $\mathcal{I}_{i_1}^1\mathcal{I}_{i_2}^2\cdots\mathcal{I}_{i_{n_\phi}}^{n_\phi}\neq \mathcal{I}_{{\eta}+1}^1\mathcal{I}_{{\eta}+1}^2\cdots\mathcal{I}_{{\eta}+1}^{n_\phi}$ for $i_j\in\lbrace1,\cdots,{\eta}+1\rbrace$.
Hence, we can additionally remove the $({\eta}+1)^{n_\phi}-1$ columns corresponding to these actions to obtain the reduced matrix game.
\end{proof}

\subsection{Existence of a MTD strategy (Problem~\ref{prob:ExistenceOfBayesNashEquilibrium})}
We now formulate a necessary and sufficient condition for the existence of a MTD strategy for the defender in the Bayesian game $\mathcal{M}$ according to Definition~\ref{def:MovingTargetDefense}.
\begin{theorem}
    \label{thm:NecessaryAndSufficientConditionForBayesNashEquilibrium}
    A moving target defense strategy exists if, and only if,
    \begin{align}
        \label{eq:NecessaryAndSufficientConditionForBayesNashEquilibrium}
        \frac{c_\text{F}}{\tau_{m-1}}\leq \frac{c_\text{F}}{\tau_m}+\sum_{j=1}^{n_\phi}\pi_j\mathcal{I}_{m}^j.
    \end{align}
\end{theorem}

\begin{proof}
    First, assume that \eqref{eq:NecessaryAndSufficientConditionForBayesNashEquilibrium} does not hold. Then \eqref{eq:DominatedStrategiesInducedMatrixGame} is fulfilled with $l=m$ and ${\eta}=m-1$. Hence, we can reduce the induced matrix game to a $1\times 1$ matrix game (see Lemma~\ref{lem:DominatedStrategiesInducedMatrixGame}), which has a pure strategy equilibrium.
    Therefore, the original Bayesian game, $\mathcal{M}$, has a unique and pure Bayesian Nash equilibrium, such that no MTD strategy exists.

    Next, we show that there exists at least one Nash equilibrium where the defender plays an MTD strategy if \eqref{eq:NecessaryAndSufficientConditionForBayesNashEquilibrium} holds. 
    Since the induced matrix game is a finite matrix game, we know that there exists at least one Nash equilibrium and equivalently at least one Bayesian Nash equilibrium for the original game, $\mathcal{M}$.
    For the Nash equilibrium to be a pure strategy Nash equilibrium, each player needs to play a best response to the other player's best response.
    Assume that $(\tau_l,\mathcal{I}_{i_1}^1\mathcal{I}_{i_2}^2\cdots\mathcal{I}_{i_{n_\phi}}^{n_\phi})$ is a pure strategy Nash equilibrium, then according to Lemma~\ref{lem:BestResponses} the following needs to be fulfilled
	\begin{align*}
       \mathcal{I}_{i_1}^1\mathcal{I}_{i_2}^2\cdots\mathcal{I}_{i_{n_\phi}}^{n_\phi}&\in  b_\text{A}(\tau_l),\\
       \tau_l&\in  b_\text{D}(\mathcal{I}_{i_1}^1\mathcal{I}_{i_2}^2\cdots\mathcal{I}_{i_{n_\phi}}^{n_\phi}),
    \end{align*}    
    i.e., each player's action is a best response to the other player's best response. 
    Comparing the first equation with the attacker's best response \eqref{eq:AttackerBestResponseInducedMatrixGame}, we see that in a pure Nash equilibrium $i_1=i_2=\cdots=i_{n_\phi}=l$.
    With \eqref{eq:DefenderBestResponseInducedMatrixGame}, we determine that the best response of the defender is
    \begin{align}
        \label{eq:BestResponsesInExistenceProof}
        \resizebox{0.99\hsize}{!}{
        $b_\text{D}(\mathcal{I}_{l}^1\mathcal{I}_{l}^2\cdots\mathcal{I}_{l}^{n_\phi})=\begin{cases}\lbrace \tau_{l-1} \rbrace &\mathrm{if}\  \frac{c_\text{F}}{\tau_{l-1}}< \frac{c_\text{F}}{\tau_m}+\sum_{j=1}^{n_\phi}\pi_j\mathcal{I}_{l}^j,\\
        \lbrace\tau_{l-1},\tau_{m}\rbrace &\mathrm{if}\  \frac{c_\text{F}}{\tau_{l-1}}= \frac{c_\text{F}}{\tau_m}+\sum_{j=1}^{n_\phi}\pi_j\mathcal{I}_{l}^j,\\
        \lbrace \tau_m\rbrace &\mathrm{otherwise}.
        \end{cases}$%
        }
    \end{align}
    To have a pure Nash equilibrium we need $l=m$. We observe that there cannot be a pure Nash equilibrium if \eqref{eq:NecessaryAndSufficientConditionForBayesNashEquilibrium} holds with inequality such that all equilibria are moving target defenses.
    However, if \eqref{eq:NecessaryAndSufficientConditionForBayesNashEquilibrium} holds with equality, the best response of the defender to $\mathcal{I}_{m}^1\mathcal{I}_{m}^2\cdots\mathcal{I}_{m}^{n_\phi}$ can be both $\tau_{m-1}$ and $\tau_{m}$. 
    Hence, in this case there exists a pure strategy Nash equilibrium in the induced matrix game.
    Next, we show that a moving target defense equilibrium strategy exists as well in this case.
    
    First, note that if \eqref{eq:NecessaryAndSufficientConditionForBayesNashEquilibrium} holds with equality then $\tau_i$ is strictly dominated by $\tau_m$ for all $i\in\lbrace1,\ldots,m-2\rbrace$, such that we can reduce the induced matrix game to a  $2\times (m^{n_{\phi}}-(m-1)^{n_{\phi}}+1)$ matrix game.
    Further, from \eqref{eq:BestResponsesInExistenceProof} we see that any distribution over $\tau_{m-1}$ and $\tau_m$ is a best response to the attack strategy $\mathcal{I}_m^1\mathcal{I}_{m}^2\cdots\mathcal{I}_m^{n_\phi}$.
    If we can show that $\mathcal{I}_m^1\mathcal{I}_{m}^2\cdots\mathcal{I}_m^{n_\phi}$ is also a best response to at least one distribution over $\tau_{m-1}$ and $\tau_m$ then we have found a Bayesian Nash equilibrium, which fulfills Definition~\ref{def:MovingTargetDefense}.
    By multiplying the attacker's payoff matrix in the reduced matrix game from the left with the distribution over $\tau_{m-1}$ and $\tau_m$, we determine that the expected payoff for playing $\mathcal{I}_{i_1}^1\mathcal{I}_{i_2}^2\cdots\mathcal{I}_{i_{\phi}}^{n_{\phi}}$ is 
    \begin{align*}
        \sum_{j=1}^{n_\phi}\pi_j\mathbf{1}_{\lbrace i_j\leq m-1\rbrace}\mathcal{I}_{i_j}^{j}+p_{m}\sum_{j=1}^{n_\phi}\pi_j\mathbf{1}_{\lbrace i_j=m\rbrace}\mathcal{I}_{m}^{j},
    \end{align*}
    where $i_j\in\lbrace 1,\ldots,m\rbrace$ for all $j\in\lbrace 1,\ldots,n_{\phi}\rbrace$ and $p_{m}$ is the probability of choosing $\tau_m$.
    Note that $\mathcal{I}_m^1\mathcal{I}_{m}^2\cdots\mathcal{I}_m^{n_\phi}$ is a best response to the mixed strategy of the defender, if the expected payoff for choosing $\mathcal{I}_m^1\mathcal{I}_{m}^2\cdots\mathcal{I}_m^{n_\phi}$ is greater than or equal to all other expected payoffs the attacker could receive, i.e., 
    \begin{align*}
        \sum_{j=1}^{n_\phi}\pi_j\mathbf{1}_{\lbrace i_j\leq m-1\rbrace}\mathcal{I}_{i_j}^{j}+p_{m}\sum_{j=1}^{n_\phi}\pi_j\mathbf{1}_{\lbrace i_j=m\rbrace}\mathcal{I}_{m}^{j}\leq p_{m}\sum_{j=1}^{n_\phi}\pi_j\mathcal{I}_{m}^{j}
    \end{align*}
    for all $i_j$ such that $\sum_{j=1}^{n_{\phi}}\mathbf{1}_{\lbrace i_j=m\rbrace}<n_{\phi}$.
    Hence, the attacker prefers to play $\mathcal{I}_m^1\mathcal{I}_{m}^2\cdots\mathcal{I}_m^{n_\phi}$ if the defender chooses 
    \begin{align*}
        p_m\in\left[\max_{i_1,\ldots,i_{n_{\phi}}}\frac{\sum_{j=1}^{n_{\phi}}\mathbf{1}_{\lbrace i_j=m-1\rbrace}\pi_j\mathcal{I}^j_{m-1}}{\sum_{j=1}^{n_{\phi}}\mathbf{1}_{\lbrace i_j = m-1\rbrace}\pi_j\mathcal{I}^j_{m}},1\right),
    \end{align*}
    where we used that $\mathcal{I}_{i_j}<\mathcal{I}_{m-1}$ for all $i_j<m-1$.
    This shows us that there there are infinitely many Nash equilibria in the induced matrix game where the defender uses a MTD strategy according to Definition~\ref{def:MovingTargetDefense} if \eqref{eq:NecessaryAndSufficientConditionForBayesNashEquilibrium} holds with equality.
    
    Therefore, we conclude that a moving target defense strategy according to Definition~\ref{def:MovingTargetDefense} exists if, and only if, \eqref{eq:NecessaryAndSufficientConditionForBayesNashEquilibrium} holds.
\end{proof}

\begin{remark}
\label{rmk:OneAttackerType}
If $\pi_i=1$ for some $i$, i.e., we have only one attacker type, the condition in Theorem~\ref{thm:NecessaryAndSufficientConditionForBayesNashEquilibrium} simplifies to the condition for the existence of a MTD from \cite{UmsonstCDC20}.
\end{remark}
\begin{remark}
In case \eqref{eq:NecessaryAndSufficientConditionForBayesNashEquilibrium} holds with equality, we can introduce an attacker type that has zero payoff as mentioned in Remark~\ref{rem:DummyAttackersMakesSense} with a prior of $\epsilon>0$ and subtract $\epsilon$ from one of the priors $\pi_j$. 
This will lead to \eqref{eq:NecessaryAndSufficientConditionForBayesNashEquilibrium} holding with inequality.
Hence, \eqref{eq:NecessaryAndSufficientConditionForBayesNashEquilibrium} can always be turned into an inequality by an arbitrarily small change in the priors.
\end{remark}

\section{Computing a MTD strategy (Problem~\ref{prob:DetermineBayesNashEquilibrium})}
In this section, we look into computing a MTD strategy. First, we investigate the general case and formulate a linear program to compute MTD strategies.
Second, we investigate the special case $n_{\phi}=1$ and provide a closed-form solution for computing a MTD strategy. 
\subsection{General case}
Finding Nash equilibria of a finite matrix game leads to a bilinear optimization problem as shown in \cite{DynamicNoncooperativeGameTheory}.
For Bayesian Nash equilibria, we can adopt the optimization problem in \cite{SolutionForBayesianGames} to obtain the following bilinear optimization problem
\begin{equation}
    \label{eq:BilinearOptForMTD}
    \begin{aligned}
        &\min_{p,q_\phi,\bar{\mathfrak{c}},\bar{\mathfrak{p}}(\phi)} \resizebox{0.8\hsize}{!}{%
        $p^T\left(\sum_{\phi=1}^{n_{\phi}}\pi_\phi\left(\Omega(\phi)-\Upsilon(\phi)\right)q_\phi\right)+\bar{\mathfrak{c}}-\sum_{\phi=1}^{n_{\phi}}\pi_\phi \bar{\mathfrak{p}}(\phi)$%
        }\\
        &\begin{aligned}
             \mathrm{s.t.}\quad &p^T1_m=1,\ q_\phi^T1_m=1,\ p\geq 0,\ q_\phi \geq 0,\\
            &\sum_{\phi=1}^{n_\phi}\pi_\phi\Omega(\phi) q_\phi \geq -\bar{\mathfrak{c}}1_m,\\ &-\Upsilon^T(\phi)p\geq \bar{\mathfrak{p}}(\phi)1_m,\ \phi\in\lbrace 1,\ldots,n_{\phi}\rbrace.
        \end{aligned}
    \end{aligned}
\end{equation}

Recall that the elements of the matrices $\Omega(\phi)$ and $\Upsilon(\phi)$ are $\Omega_{i,j}(\phi)=\frac{c_\text{F}}{\tau_i}+\mathbf{1}_{\lbrace j\leq i\rbrace}\mathcal{I}^\phi_j$ and $\Upsilon_{i,j}(\phi)=\mathbf{1}_{\lbrace j\leq i\rbrace}\mathcal{I}^\phi_j$, respectively.
Here, $q_\phi$ is the mixed strategy for the attacker with type $\phi$ and $\bar{\mathfrak{p}}(\phi)$ is its average payoff, while $p$ is the mixed strategy of the defender and $\bar{\mathfrak{c}}$ is its average cost.

\begin{proposition}
\label{prop:LinearProgForBayesMTD}
Assume that the condition of Theorem~\ref{thm:NecessaryAndSufficientConditionForBayesNashEquilibrium} holds, and thus a MTD exists. A MTD strategy can then be computed by solving the linear program,
\begin{equation}
\label{eq:LinearOptForBayesMTD}
    \begin{aligned}
        &\min_{p,q_\phi,\bar{\mathfrak{c}},\bar{\mathfrak{p}}(\phi)} p^T\gamma+\bar{\mathfrak{c}}-\sum_{\phi=1}^{n_{\phi}}\pi_\phi \bar{\mathfrak{p}}(\phi)\\
        &
        \begin{aligned}
            \mathrm{s.t.}\quad &p^T1_m=1,\ q_\phi^T1_m=1,\ p\geq 0,\ q_\phi\geq 0,\\
            &\gamma+\sum_{\phi=1}^{n_\phi}\pi_\phi\Upsilon(\phi) q_\phi\geq -\bar{\mathfrak{c}}1_m,\\ &-\Upsilon^T(\phi)p\geq \bar{\mathfrak{p}}(\phi)1_m,\ \phi\in\lbrace 1,\ldots,n_{\phi}\rbrace,
        \end{aligned}
    \end{aligned}
\end{equation}
where the $l$th element of the $m$-dimensional vector $\gamma$ is $\frac{c_\text{F}}{\tau_l}$.
\end{proposition}
\begin{proof}
Due to the special structure of $\Omega(\phi)$, we note that the $l$th element of $\Omega(\phi)q_\phi$ equals $\frac{c_{\text{F}}}{\tau_l}+\sum_{j=1}^l q_{\phi,j}\mathcal{I}_j^\phi$, such that $\Omega(\phi)q_\phi=\gamma+\Upsilon(\phi)q_\phi$.
Inserting that in the objective function and the constraints of \eqref{eq:BilinearOptForMTD} leads to the optimization problem in \eqref{eq:LinearOptForBayesMTD}, where we further used that $\sum_{\phi=1}^{n_{\phi}}\pi_\phi=1$.
\end{proof}
The optimization problem in \eqref{eq:LinearOptForBayesMTD} is a convex linear program and therefore, we are guaranteed to find the global optimum.
This has an advantage over directly solving \eqref{eq:BilinearOptForMTD}, where we may get stuck in a local optimum.
\subsection{Special case: $n_{\phi}=1$}
\label{sec:ClosedFormSolution}
Next, we provide a closed-form solution to  Problem~\ref{prob:DetermineBayesNashEquilibrium}, when the defender faces only one attacker type, i.e. $n_{\phi}=1$, which is the problem we mentioned in Remark~\ref{rmk:OneAttackerType}.
For ease of notation, we will omit the superscript for the attacker type.

For $n_{\phi}=1$ the matrix representation of the Bayesian Nash equilibrium definition in \eqref{eq:BayesNashEquilibriumGeneral} simplifies to the definition of the Nash equilibrium
\begin{equation}
    \label{eq:NashEquilibrium}
    \begin{aligned}
        (p^*)^T\Omega q^*&\leq p^T\Omega q^* &\forall p\in\Delta_p,\\
        (p^*)^T\Upsilon q^*&\geq (p^*)^T\Upsilon q &\forall q\in\Delta_q.
    \end{aligned}
\end{equation}

Let $\mathcal{Q}$ denote the support of the attacker's mixed strategy, i.e., if $i\in\mathcal{Q}$ then the attacker chooses $\mathcal{I}_i$ with a nonzero probability $q_i>0$ and if $i\not\in\mathcal{Q}$ then $q_i=0$. 
The support for the mixed strategy of the defender is defined in a similar way and is denoted by $\mathcal{P}$ with probabilities $p_i$.
In the following, we investigate one mixed strategy for the defender and one for the attacker and show how the support of the best response of the attacker, respectively defender, has to look like.
We then use this to define a mixed strategy Nash equilibrium, which represents a MTD.

\begin{lemma}
\label{lem:ClosedFormNashForAttacker}
If the attacker fixes $i, 1<i<m$, and uses the mixed strategy
\begin{align}
	\label{eq:ClosedFormNashForAttacker}
        q_{j}=\begin{cases}\frac{c_\text{F}}{\mathcal{I}_{j}}\left(\frac{1}{\tau_{j-1}}-\frac{1}{\tau_{j}}\right), &\mathrm{if}\ j\in\lbrace i+1,\ldots,m \rbrace,\\
        1-\sum_{l=i+1}^m q_l, &\mathrm{if}\ j=i,\\ 
        0, &\mathrm{otherwise},
        \end{cases},
    \end{align}
   where
    \begin{align}
    \label{eq:ConditionForQi}
        0\leq q_i<\max\left(1,\frac{c_{\text{F}}}{\mathcal{I}_i}\left(\frac{1}{\tau_{i-1}}-\frac{1}{\tau_i}\right)\right),
    \end{align}
     then $\mathcal{P}\subseteq\lbrace i,\ldots,m\rbrace$ needs to hold for the support of the defender's best response.
\end{lemma}
\begin{proof}
First, note that $q_j>0$ for $j\in\lbrace i+1,\ldots,m \rbrace$, since $\tau_j>\tau_i$ if $j>i$. 
Further, if $q_i\in[0,1)$ we see that the mixed strategy $q$ given by \eqref{eq:ClosedFormNashForAttacker} is a proper probability distribution.

Next, we look at the possible best responses of the defender to this strategy.
The average cost of the defender is given by
\begin{align*}
    p^T\Omega q&=p^T\gamma+p^T\Upsilon q\\
    &=p_{1:i-1}^T\tilde{\gamma}+p_{i:m}^T\hat{\gamma}+p^T\begin{bmatrix} \Upsilon_{11} & \Upsilon_{12}\\\Upsilon_{21} & \Upsilon_{22}
    \end{bmatrix}q\\
    &=p_{1:i-1}^T\tilde{\gamma}+p_{i:m}^T\hat{\gamma}+\begin{bmatrix}p_{1:i-1}^T &p_{i:m}^T\end{bmatrix}\begin{bmatrix}0\\\Upsilon_{22}q_{i:m}
    \end{bmatrix}\\
    &=p_{1:i-1}^T\tilde{\gamma}+p_{i:m}^T\hat{\gamma}+p_{i:m}^T\Upsilon_{22}q_{i:m},
\end{align*}
where $\tilde{\gamma}^T=[\frac{c_\text{F}}{\tau_1},\ldots,\frac{c_\text{F}}{\tau_{i-1}}]$, $\hat{\gamma}^T=[\frac{c_\text{F}}{\tau_i},\ldots,\frac{c_\text{F}}{\tau_m}]$, and
\begin{align*}
    \Upsilon_{22}=\begin{bmatrix}
         \mathcal{I}_i & 0 & \cdots & 0\\
          \mathcal{I}_i & \mathcal{I}_{i+1} & \cdots & 0\\
          \vdots &\vdots &\ddots &\vdots\\
          \mathcal{I}_i & \mathcal{I}_{i+1} & \cdots & \mathcal{I}_m
    \end{bmatrix}.
\end{align*}
With that we can determine that
\begin{align*}
    \Upsilon_{22}q_{i:m}=\begin{bmatrix}
        q_i\mathcal{I}_i\\
        \sum_{j=i}^{i+1}q_j\mathcal{I}_j\\
        \vdots \\
        \sum_{j=i}^{m}q_j\mathcal{I}_j
    \end{bmatrix}=\begin{bmatrix}
        q_i\mathcal{I}_i\\
        q_i\mathcal{I}_i+\frac{c_{\text{F}}}{\tau_i}-\frac{c_{\text{F}}}{\tau_{i+1}}\\
        \vdots\\
        q_i\mathcal{I}_i+\frac{c_{\text{F}}}{\tau_i}-\frac{c_{\text{F}}}{\tau_{m}}
    \end{bmatrix},
\end{align*}
where we used \eqref{eq:ClosedFormNashForAttacker} for the values of $q_j$ for $j> i$.
This leads to the following average cost of the defender
\begin{align*}
    p^T\Omega q&=p_{1:i-1}^T\tilde{\gamma}+p_{i:m}^T\hat{\gamma}+p_{i:m}^T\Upsilon_{22}q_{i:m}\\
    &=p_{1:i-1}^T\tilde{\gamma}+p_{i:m}^T\hat{\gamma}+\big(q_i\mathcal{I}_i+\frac{c_{\text{F}}}{\tau_i}\big)\sum_{j=i}^{m}p_j-p_{i:m}^T\hat{\gamma}\\
    &=p_{1:i-1}^T\tilde{\gamma}+\big(q_i\mathcal{I}_i+\frac{c_{\text{F}}}{\tau_i}\big)\sum_{j=i}^{m}p_j.
\end{align*}
Now assume $\mathcal{P}\subseteq\lbrace i,\ldots,m\rbrace$, then $p_{1:i-1}=0$ and $\sum_{j=i}^m p_j=1$, such that the average cost turns into $p^T\Omega q=q_i\mathcal{I}_i+\frac{c_{\text{F}}}{\tau_i}$, which shows us that the defender is indifferent among its action, as it obtains the same average cost no matter how the distribution $p_{i:m}$ is chosen.

Now let $\mathcal{P}\not\subseteq\lbrace i,\ldots,m\rbrace$, then $\sum_{j=i}^m p_j=1-p_{1:i-1}^T1_{i-1}$, such that the average cost becomes
\begin{align*}
     p^T\Omega q=q_i\mathcal{I}_i+\frac{c_{\text{F}}}{\tau_i}+p_{1:i-1}^T\bigg(\tilde{\gamma}-\big(q_i\mathcal{I}_i+\frac{c_{\text{F}}}{\tau_i}\big)1_{i-1}\bigg).
\end{align*}
The defender chooses $\mathcal{P}\not\subseteq\lbrace i,\ldots,m\rbrace$ if, and only if,
    \begin{align*}
        q_i\mathcal{I}_i+\frac{c_{\text{F}}}{\tau_i}+p_{1:i-1}^T\bigg(\tilde{\gamma}-\big(q_i\mathcal{I}_i+\frac{c_{\text{F}}}{\tau_i}\big)1_{i-1}\bigg)&\leq q_i\mathcal{I}_i+\frac{c_{\text{F}}}{\tau_i}\\
        \Leftrightarrow p_{1:i-1}^T\bigg(\tilde{\gamma}-\big(q_i\mathcal{I}_i+\frac{c_{\text{F}}}{\tau_i}\big)1_{i-1}\bigg)&\leq 0.
    \end{align*}
Since the elements of both $p_{1:i-1}$ and $\tilde{\gamma}$ are positive and 
$\frac{c_{\text{F}}}{\tau_1}>\frac{c_{\text{F}}}{\tau_2}>\ldots>\frac{c_{\text{F}}}{\tau_{i-1}}$, we obtain the following necessary condition for choosing $\mathcal{P}\not\subseteq\lbrace i,\ldots,m\rbrace$,
\begin{align*}
    \frac{c_{\text{F}}}{\tau_{i-1}}\leq q_i\mathcal{I}_i+\frac{c_{\text{F}}}{\tau_i}.
\end{align*}
Therefore, if 
\begin{align*}
    \frac{c_{\text{F}}}{\tau_{i-1}}> q_i\mathcal{I}_i+\frac{c_{\text{F}}}{\tau_i},
\end{align*}
then $p_{1:i-1}^T\bigg(\tilde{\gamma}-\big(q_i\mathcal{I}_i+\frac{c_{\text{F}}}{\tau_i}\big)1_{i-1}\bigg)> 0$ and the defender chooses $\mathcal{P}\subseteq\lbrace i,\ldots,m\rbrace$ to minimize its cost. 
Reformulating this inequality gives us the upper bound
\begin{align*}
    q_i<\frac{c_{\text{F}}}{\mathcal{I}_i}\left(\frac{1}{\tau_{i-1}}-\frac{1}{\tau_i}\right).
\end{align*}
Note that if $\tau_{i-1}$ is strictly dominated by $\tau_{i}$, this upper bound is larger than $1$ and therefore automatically fulfilled, if $q_i\in[0,1)$.
However, if $\tau_{i-1}$ is \emph{not} strictly dominated by $\tau_i$ then both
\begin{align*}
    \frac{c_{\text{F}}}{\tau_{i-1}}>\frac{c_{\text{F}}}{\tau_{i}}\ \mathrm{and}\ \frac{c_{\text{F}}}{\tau_{i-1}}\leq\frac{c_{\text{F}}}{\tau_{i}}+\mathcal{I}_i
\end{align*}
hold, which means there exists $\rho_i\in(0,1]$ such that
\begin{align*}
    \frac{c_{\text{F}}}{\tau_{i-1}}=\frac{c_{\text{F}}}{\tau_{i}}+\rho_i\mathcal{I}_i
\end{align*}
holds and we can determine $\rho_i$ as
\begin{align*}
    \rho_i=\frac{c_{\text{F}}}{\mathcal{I}_i}\left(\frac{1}{\tau_{i-1}}-\frac{1}{\tau_i}\right).
\end{align*}

Hence, if $q$ is chosen according to \eqref{eq:ClosedFormNashForAttacker} such that $q_i$ fulfills \eqref{eq:ConditionForQi}, the defender chooses $\mathcal{P}\subseteq\lbrace i,\ldots,m\rbrace$.
\end{proof}

Lemma~\ref{lem:ClosedFormNashForAttacker} shows us the support of best responses for the defender to the attack strategy \eqref{eq:ClosedFormNashForAttacker}, however, it still leaves the open question how to choose $i$ such that \eqref{eq:ConditionForQi} is fulfilled.

\begin{lemma}
    \label{lem:UniqueIndexForQ}
    There exists a unique index $i=i^*\in(1,m-1)$ so that \eqref{eq:ClosedFormNashForAttacker} is a proper probability distribution and \eqref{eq:ConditionForQi} holds if ${1-\sum_{j=2}^m\frac{c_\text{F}}{\mathcal{I}_{j}}\left(\frac{1}{\tau_{j-1}}-\frac{1}{\tau_{j}}\right)<0}$ and $\frac{c_{\text{F}}}{\mathcal{I}_m}\left(\frac{1}{\tau_{m-1}}-\frac{1}{\tau_{m}}\right)<1$.
\end{lemma}
\begin{proof}
    From \eqref{eq:ClosedFormNashForAttacker}, we obtain that $q_i=1-\sum_{j=i+1}^m\frac{c_\text{F}}{\mathcal{I}_{j}}\left(\frac{1}{\tau_{j-1}}-\frac{1}{\tau_{j}}\right)$, which is strictly decreasing as $i$ decreases.
    Since $\frac{c_{\text{F}}}{\mathcal{I}_m}\left(\frac{1}{\tau_{m-1}}-\frac{1}{\tau_{m}}\right)<1$ holds, we know that $q_i\in(0,1)$ for $i=m-1$.
    Further, since ${1-\sum_{j=2}^m\frac{c_\text{F}}{\mathcal{I}_{j}}\left(\frac{1}{\tau_{j-1}}-\frac{1}{\tau_{j}}\right)<0}$ we know that $q_i<0$ for $i=1$.
    Hence, there exists an $i=i^*>1$ such that $q_i\geq 0$ while for $i<i^*$ we have $q_i<0$, such that the mixed strategy in \eqref{eq:ClosedFormNashForAttacker} is not a proper probability distribution and therefore not a valid strategy.
    For $i=i^*$ we can, therefore, show that
       \begin{multline*}
        1-\sum_{l=i^*}^m\frac{c_{\text{F}}}{\mathcal{I}_l}\left(\frac{1}{\tau_{l-1}}-\frac{1}{\tau_{l}}\right)< 0\\
        \Leftrightarrow 1-\sum_{l=i^*+1}^m\frac{c_{\text{F}}}{\mathcal{I}_l}\left(\frac{1}{\tau_{l-1}}-\frac{1}{\tau_{l}}\right)-\frac{c_{\text{F}}}{\mathcal{I}_{i^*}}\left(\frac{1}{\tau_{i^*-1}}-\frac{1}{\tau_{i^*}}\right)< 0\\
        \Leftrightarrow q_{i^*}-\frac{c_{\text{F}}}{\mathcal{I}_{i^*}}\left(\frac{1}{\tau_{i^*-1}}-\frac{1}{\tau_{i^*}}\right)< 0\\
        \Leftrightarrow q_{i^*}< \frac{c_{\text{F}}}{\mathcal{I}_{i^*}}\left(\frac{1}{\tau_{i^*-1}}-\frac{1}{\tau_{i^*}}\right).
    \end{multline*}
    Hence, \eqref{eq:ConditionForQi} holds for $i=i^*$.
    Now assume we choose $i=j>i^*$ such that $q_i\in[0,1)$ if $i=j$ and also $q_i\geq 0$ if $i=j-1$. 
    Then we obtain that
    \begin{multline*}
        1-\sum_{l=j}^m\frac{c_{\text{F}}}{\mathcal{I}_l}\left(\frac{1}{\tau_{l-1}}-\frac{1}{\tau_{l}}\right)\geq 0\\
        \Leftrightarrow 1-\sum_{l=j+1}^m\frac{c_{\text{F}}}{\mathcal{I}_l}\left(\frac{1}{\tau_{l-1}}-\frac{1}{\tau_{l}}\right)-\frac{c_{\text{F}}}{\mathcal{I}_{j}}\left(\frac{1}{\tau_{j-1}}-\frac{1}{\tau_{j}}\right)\geq 0\\
        \Leftrightarrow q_{j}-\frac{c_{\text{F}}}{\mathcal{I}_{j}}\left(\frac{1}{\tau_{j-1}}-\frac{1}{\tau_{j}}\right)\geq 0\\
        \Leftrightarrow q_{j}\geq \frac{c_{\text{F}}}{\mathcal{I}_{j}}\left(\frac{1}{\tau_{j-1}}-\frac{1}{\tau_{j}}\right).
    \end{multline*}
    Hence, \eqref{eq:ConditionForQi} does not hold for any $i\neq i^*$. Therefore, $i=i^*$ is the smallest index for which $q_i^*\in[0,1)$ and the unique index for which \eqref{eq:ConditionForQi} holds.
\end{proof}

Similar to Lemma~\ref{lem:ClosedFormNashForAttacker} we can find a mixed strategy $p$ with support $\mathcal{P}=\lbrace i,i+1,\ldots,m \rbrace$ such that the best response of the attacker has support $\mathcal{Q}\subseteq\lbrace i,i+1,\ldots,m \rbrace$.

\begin{lemma}
\label{lem:ClosedFormNashForDefender}
If the defender uses the mixed strategy
\begin{align}
        \label{eq:ClosedFormNashForDefender}
        p_{j}=\begin{cases}1-\frac{\mathcal{I}_{i}}{\mathcal{I}_{i+1}}, & \mathrm{if}\ j=i,\\ 
        \frac{\mathcal{I}_{i}}{\mathcal{I}_{j}}-\frac{\mathcal{I}_{i}}{\mathcal{I}_{j+1}}, & \mathrm{if}\ j\in\lbrace i+1,\ldots,m-1\rbrace,\\
        \frac{\mathcal{I}_{i}}{\mathcal{I}_{m}}, & \mathrm{if}\ j=m,\\
        0, &\mathrm{otherwise},\end{cases}
\end{align}
then the support of the attacker's best response needs to satisfy $\mathcal{Q}\subseteq\lbrace i,i+1,\ldots,m \rbrace$.
\end{lemma}
\begin{proof}
First note that since $\mathcal{I}_j>\mathcal{I}_i$ if $j>i$, we see that each $p_i\in (0,1)$. Furthermore, we can verify that ${\sum_{j=1}^m p_j=1}$. Hence, the mixed strategy described by \eqref{eq:ClosedFormNashForDefender} is a proper probability distribution.

Next, we look at the possible best responses of the attacker to this strategy. 
The average cost of the attacker is
\begin{align*}
    p^T\Upsilon q&=p^T\begin{bmatrix} \Upsilon_{11} & \Upsilon_{12}\\\Upsilon_{21} & \Upsilon_{22}
    \end{bmatrix}q\\
    &=p_{i:m}^T\Upsilon_{21}q_{1:i-1}+p_{i:m}^T\Upsilon_{22}q_{i:m}\\
    &=\sum_{j=1}^{i-1}q_j\mathcal{I}_j+p_{i:m}^T\Upsilon_{22}q_{i:m},
\end{align*}
where we used that $\Upsilon_{12}=0$ and $\Upsilon_{21}=1_{i-1}[\mathcal{I}_1,\mathcal{I}_2,\ldots,\mathcal{I}_{i-1}]$.
Due to the chosen $p$, we obtain that $p_{i:m}^T\Upsilon_{22}=1_{m-i+1}^T\mathcal{I}_i$, which results in the following average cost
\begin{align*}
     p^T\Upsilon q&=\sum_{j=1}^{i-1}q_j\mathcal{I}_j+\mathcal{I}_i\sum_{l=i}^m q_l\\
     &=\sum_{j=1}^{i-1}q_j\mathcal{I}_j+\mathcal{I}_i\left(1-\sum_{j=1}^{i-1}q_j\right)\\
     &=\mathcal{I}_i+\sum_{j=1}^{i-1}q_j(\mathcal{I}_j-\mathcal{I}_i)\leq \mathcal{I}_i.
\end{align*}
The upper bound comes from the fact that $\mathcal{I}_j<\mathcal{I}_i$ if $j\in\lbrace 1,\ldots,i-1\rbrace$ (see Assumption~\ref{assum:PayoffFunctionOfAttacker}).
Hence, we see that the best response of the attacker to the defender's strategy $p$ is any mixed strategy $q$ with support $\mathcal{Q}\subseteq\lbrace i,i+1,\ldots,m \rbrace$.
\end{proof}

One notable difference between the results given in Lemma~\ref{lem:ClosedFormNashForAttacker} and Lemma~\ref{lem:ClosedFormNashForDefender} is that the defender's mixed strategy $p$ given by \eqref{eq:ClosedFormNashForDefender} is valid for all $i$, while $q$ given by \eqref{eq:ClosedFormNashForAttacker} has the additional constraint \eqref{eq:ConditionForQi}.
However, Lemma~\ref{lem:UniqueIndexForQ} shows us that under a certain condition there exists a unique index $i$ for which \eqref{eq:ConditionForQi} holds.
Next, we show that for a specific choice of $i$ the strategies \eqref{eq:ClosedFormNashForAttacker} and \eqref{eq:ClosedFormNashForDefender} form a mixed strategy Nash equilibrium.

\begin{theorem}
    \label{thm:ClosedFormSolutionForNashEquilibrium}
    Let $i=i^*\in[1,m-1]$ be the \emph{smallest} index for which $q_i^*\in[0,1)$ in \eqref{eq:ConditionForQi} holds such that $q^*$ is a proper probability distribution.
    The mixed strategies $p^*$ and $q^*$ given by \eqref{eq:ClosedFormNashForDefender} and \eqref{eq:ClosedFormNashForAttacker}, respectively, form a mixed strategy Nash equilibrium such that $p^*$ is a MTD if, and only if, $\frac{c_{\text{F}}}{\mathcal{I}_m}\left(\frac{1}{\tau_{m-1}}-\frac{1}{\tau_{m}}\right)\leq 1$.
\end{theorem}
\begin{proof}
    We begin by noting that if $\frac{c_{\text{F}}}{\mathcal{I}_m}\left(\frac{1}{\tau_{m-1}}-\frac{1}{\tau_{m}}\right)> 1$, then $q_m>1$ such that \eqref{eq:ClosedFormNashForAttacker} is not a proper probability distribution.
    Next, we show that if $\frac{c_{\text{F}}}{\mathcal{I}_m}\left(\frac{1}{\tau_{m-1}}-\frac{1}{\tau_{m}}\right)\leq 1$ then $i^*\in[1,m-1]$ such that $q_{i^*}\in[0,1)$ exists. 
    For this, we need to consider three different cases.
    In the first case, we assume that ${1-\sum_{j=2}^m\frac{c_\text{F}}{\mathcal{I}_{j}}\left(\frac{1}{\tau_{j-1}}-\frac{1}{\tau_{j}}\right)<0}$ and $\frac{c_\text{F}}{\mathcal{I}_m}(\frac{1}{\tau_{m-1}}-\frac{1}{\tau_m})<1$.
    Then Lemma~\ref{lem:UniqueIndexForQ} shows us that a unique $i^*\in(1,m-1)$ exists, for which \eqref{eq:ConditionForQi} holds and $i^*$ is also the smallest index for which $q_i\in[0,1)$.
    In the second case, we assume that ${1-\sum_{j=2}^m\frac{c_\text{F}}{\mathcal{I}_{j}}\left(\frac{1}{\tau_{j-1}}-\frac{1}{\tau_{j}}\right)\geq 0}$ and $\frac{c_\text{F}}{\mathcal{I}_m}(\frac{1}{\tau_{m-1}}-\frac{1}{\tau_m})<1$. Then $i^*=1$ guarantees that $q_i\in[0,1)$.
    Furthermore, $i^*=1$ is also the smallest index in this case for which $q_i\in[0,1)$.
    In the third case, we assume that $\frac{c_\text{F}}{\mathcal{I}_m}(\frac{1}{\tau_{m-1}}-\frac{1}{\tau_m})=1$, which shows us that $i^*=m-1$ is the smallest index in this case for which $q_i\in[0,1)$.
    Hence, there exists a unique index $i^*\in[1,m-1]$, which is the smallest index such that $q_{i}\in[0,1)$ if, and only, if $\frac{c_{\text{F}}}{\mathcal{I}_m}\left(\frac{1}{\tau_{m-1}}-\frac{1}{\tau_{m}}\right)\leq 1$.
    
    Next, if we use $q^*$ with $i=i^*$, we see that the support of the defender's best response needs to fulfil $\mathcal{P}\subseteq\lbrace i^*, i^*+1, \ldots, m\rbrace$, which is fulfilled when $p^*$ is used.
    Therefore, $p^*$ is a best response to $q^*$.
    
    Finally, if we use $p^*$ with $i=i^*$, we see that the support of the defender's best response needs to fulfill $\mathcal{Q}\subseteq\lbrace i^*, i^*+1, \ldots, m\rbrace$, which is fulfilled when $q^*$ is used.
    Therefore, $q^*$ is a best response to $p^*$.
    Hence, $p^*$ and $q^*$ form a mixed strategy Nash equilibrium and $p^*$ is a MTD according to Definition~\ref{def:MovingTargetDefense}.
\end{proof}

Theorem~\ref{thm:ClosedFormSolutionForNashEquilibrium} presents one optimal solution to the optimization problem \eqref{eq:LinearOptForBayesMTD} when $n_\phi=1$. 
Our numerical experiments in Section~\ref{sec:SimulationsForClosedLoopSolution} show that the optimal solution obtained by solving \eqref{eq:LinearOptForBayesMTD} coincides with the equilibrium proposed in Theorem~\ref{thm:ClosedFormSolutionForNashEquilibrium}.

\begin{remark}
Note that if $\frac{c_\text{F}}{\mathcal{I}_m}(\frac{1}{\tau_{m-1}}-\frac{1}{\tau_m})=1$, then we have $i^*=m-1$ and $q_{i^*}=0$ such that the attacker plays a pure strategy. 
This means there are at least two Nash equilibria, one given by $p^*$ and $\mathcal{I}_m$, and one given by $\tau_m$ and $\mathcal{I}_m$.
In this case, our matrix game is degenerate but there still exists an MTD according to Definition~\ref{def:MovingTargetDefense}.
\end{remark}

\begin{remark}
\label{rem:DummyAttackerForClosedFormSolution}
As already mentioned in Remark~\ref{rem:DummyAttackersMakesSense}, it is often reasonable to include an attacker type that has zero payoff for all trajectories $a$ such that the defender does also consider the case without an attacker when choosing the threshold.
For the closed-form solution presented in this section, we are able to replace $\mathcal{I}_j$ by $\pi_1\mathcal{I}_j$ to take the attacker type with zero payoff into account, where $\pi_1\in(0,1]$ is the probability that the attack is happening.
Interestingly, this modification does not change the defender's equilibrium strategy \eqref{eq:ClosedFormNashForDefender}.
However, it does lead to a larger $i^*$ in Lemma~\ref{lem:UniqueIndexForQ} (see Section~\ref{sec:SimulationsForClosedLoopSolution}). 
This is consistent with our intuition since the attack will be less likely to happen and the defender focuses on increasing the mean time between false alarms by choosing larger thresholds.
\end{remark}
\section{Numerical Evaluation}
\label{sec:NumericalEvaluation}
For the numerical evaluation, we look at a four-tank system~\cite{QuadrupleTank}, which we linearize around the input voltage of $6\,\mathrm{V}$ and discretize with a sampling time of $0.5\,\mathrm{s}$.
We further assume that $w(k)\sim \mathcal{N}(0,0.1 I_4)$ and $v(k)\sim \mathcal{N}(0,0.01 I_2)$ and an LQG controller is used, where the LQR cost matrix for the states is $I_4$ and the cost matrix for the controller input is $I_2$, i.e., the controller input $u(k)=-Kx(k)$ minimizes the cost function $\sum_{k=0}^\infty x(k)^Tx(k)+u(k)^Tu(k)$.
The attack length is chosen to be $N=1000$ time steps.
For anomaly detection, a $\chi^2$ detector is used such that
\begin{align*}
    y_D(k+1)=r(k)^Tr(k).
\end{align*}

\subsection{Bayesian Nash equilibrium}
In this part, we solve the optimization problem \eqref{eq:LinearOptForBayesMTD} to find the equilibrium moving target defense for the defender. 
Note that we choose the defender's set of actions, the attacker type payoff functions and the factor $c_F$ in the defender's objective function for illustrative purposes of the presented MTD framework. 
In practice, the defender needs to choose $c_F$ according to its cost for false alarm and the payoff functions of the attacker types could be the result of a risk assessment as discussed in Section~\ref{sec:DiscussionOfAssumptions}.

The defender considers six thresholds, which correspond to the following average times between false alarms,
\begin{align}
    \label{eq:SetForTauInSim}
    \tau\in\lbrace10,\ 10^2,\ 10^3,\ 10^4,\ 10^5,\ 10^6\rbrace.
\end{align}
These values are chosen to cover a wide range of average times between false alarms. 
Further, we use $c_F=43200$ as the cost factor for false alarms.

The defender faces $n_\phi=5$ attacker types. 
We further assume that the attack starts at $\ubar{k}=0$.
The first attacker type is an attacker with zero payoff, i.e., $f_1(a)=0$ for all $a$.
This type represents the case where there is actually no attacker present in the system and the defender only has to consider the cost induced by the false alarms.
For the other attacker types, we use the average value of the plant's state at the end of the attack, i.e., $\bar{x}=\mathbb{E}\lbrace x(N)\rbrace$, to define the payoff. 
For attacker type $\phi\in\lbrace 2,\ldots, 5\rbrace$, we use $f_\phi(a)=|\bar{x}_{\phi-1}|^2$ as the payoff function. 
Thus, attacker type~2 attacks the water level in tank~1, attacker type~3 attacks the water level in tank~2 and so on.
Since a $\chi^2$ detector is used, we can use the results of Proposition~3 in \cite{UmsonstCDC18} to determine the attack impact for each attacker type for a given $\tau$.

We consider three different scenarios that differ in terms of their priors $\pi_\phi$.
In the first scenario, the operator assumes that it is more likely that there is no attack than that there is an attack, and thus $\pi_1=0.6$ and $\pi_\phi=0.1$ for $\phi\in\lbrace 2,\ldots,5\rbrace$.
In the second scenario, the operator assumes that there is always an attack but we want to investigate how the defender's strategy changes when the attacks are still equally likely, i.e., $\pi_1=0$ and $\pi_\phi=0.25$ for $\phi\in\lbrace 2,\ldots,5\rbrace$.
In the third scenario, there is also always an attack but this time the attacker is assumed to most likely attack the first and second state of the plant, i.e., $\pi_1=0$, $\pi_2=\pi_3=0.49$, and $\pi_4=\pi_5=0.01$.

Figure~\ref{fig:OptimalStrategiesDef} shows the equilibrium MTD of the defender for the three different scenarios.
\begin{figure}
    \centering
    \includegraphics[scale=0.55]{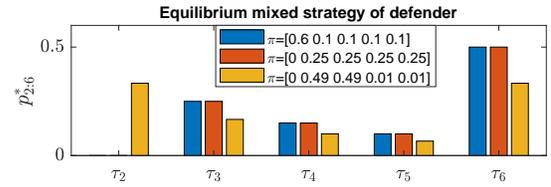}
    \caption{The plot shows the MTD of the defender for three different priors of an attack happening, where the horizontal axis are the defender actions and the vertical axis shows the probability of choosing the respective action.}
    \label{fig:OptimalStrategiesDef}
\end{figure}
Figure~\ref{fig:OptimalStrategiesAtk} shows the equilibrium mixed strategies for attacker type 2 to attacker type 5.
\begin{figure}
    \centering
    \includegraphics[scale=0.55]{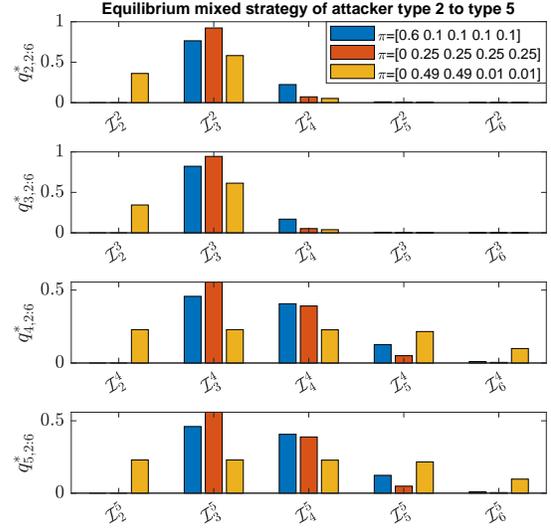}
    \caption{For three different priors, the equilibrium mixed strategies of attacker type 2 to type 5 are shown in each of the subfigures. 
    The horizontal axis shows the attacker actions while the vertical axis shows the probability of choosing the respective action.}
    \label{fig:OptimalStrategiesAtk}
\end{figure}
Using Lemma~\ref{lem:DominatedStrategiesInducedMatrixGame}, we can determine that for the first two scenarios the defender's actions $\tau_1$ and $\tau_2$ are strictly dominated and, therefore, will not be used in the Bayesian Nash equilibrium. 
In the third scenario, only $\tau_1$ is strictly dominated and $\tau_2$ is used in the Bayesian Nash equilibrium.
This means that the cost for an attack that is stealthy for $\tau_1$ is negligible compared to the cost for false alarms.
Further, this also means none of the attacker types will use the attack action corresponding to these thresholds in the respective scenarios, because the attacker wants to maximize its payoff.
Hence, $p^*_{1}=q^*_{\phi,1}=0$ for all $\phi$ and these values are not depicted in Figure~\ref{fig:OptimalStrategiesDef} and Figure~\ref{fig:OptimalStrategiesAtk} for the sake of simplicity.
Furthermore, since the attacker type 1 obtains always zero payoff it does not influence the objective value of the optimization problem in \eqref{eq:LinearOptForBayesMTD}.
Therefore, we can arbitrarily choose $q^*_{1,1:6}=\frac{1}{6}1_6$ as the equilibrium mixed strategy of attacker type 1 for all three scenarios.
This is not shown in the figures for the sake of simplicity.

From Figure~\ref{fig:OptimalStrategiesDef}, we can observe that the defender chooses $\tau_6$ with the highest probability in the MTD, while the attacker puts most of the probability weight on the attacks with a lower payoff than the payoff corresponding to $\tau_6$.
Since the attack will not receive any payoff if it is detected, this observation is reasonable and also shows that the proposed moving target defense is effectively limiting the attacker payoff.
This can be observed especially in the third scenario, where attacker type 2 and type 3 are more likely than attacker type 4 and type 5. In addition to that, the payoffs of attacker type 2 and type 3 are larger than for type 4 and type~5. 
Hence, a larger attacker payoff is more likely than in the second scenario.
Therefore, in the third scenario the defender chooses $\tau_2$ with a non-zero probability to force the attacker to remain undetected and receive a lower payoff.
So we see that in this case having more false alarms outweighs the cost of having an undetectable attack.

We can also observe from Figure~\ref{fig:OptimalStrategiesDef} that the defender has the same MTD for the first two scenarios. 
A reason for this is that in \eqref{eq:LinearOptForBayesMTD} the defender's constraints will not be affected by the attacker type with zero payoff. 
Hence, the constraint set for choosing $p$ is the same in the first and second scenario due to the uniform prior across attacker types 2 to 5.
For the different attacker types, it is interesting to see that although the defender has the same MTD in the first and second scenario, the attacker types' mixed strategies do change.
We, further, observe that the mixed strategies $q_{2}^*$ and $q_{3}^*$ for attacker type 2 and type 3, respectively, are very close and the mixed strategies $q_{4}^*$ and $q_{5}^*$ are close as well, in all three investigated scenarios. 
For example, $\|q_{2}^*-q_{3}^*\|_{\infty}$ for the first, second, and third scenario is $0.0672$, $0.0199$, and $0.0131$, respectively, while  $\|q_{4}^*-q_{5}^*\|_{\infty}$ for the first, second, and third scenario is $0.0039$, $0.0024$, and $0.0053$, respectively.

Finally, we look at an interesting  property of the MTD obtained by solving \eqref{eq:LinearOptForBayesMTD}.
The MTD strategy in the first and second scenario, where all attacker types with a nonzero impact have the same prior, is
\begin{align}
    \label{eq:MTDScenario2}
    p^*=\begin{bmatrix} 0 & 0 & 0.25 & 0.15 & 0.1 & 0.5
\end{bmatrix}^T
\end{align}
while in the third scenario the MTD is
\begin{align}
    \label{eq:MTDScenario3}
    p^*=\begin{bmatrix}0 & 0.3333 & 0.1667 & 0.1 & 0.0667 & 0.3333\end{bmatrix}^T.
\end{align}
Interestingly, both \eqref{eq:MTDScenario2} and \eqref{eq:MTDScenario3} have the structure of the proposed MTD in Theorem~\ref{thm:ClosedFormSolutionForNashEquilibrium}, although Theorem~\ref{thm:ClosedFormSolutionForNashEquilibrium} is only for the case where the defender faces one specific attacker type.
This can be explained with the structure of the payoff for each attacker type. 
We determine that $\mathcal{I}_j^{\phi}=\alpha_{\phi}J_D(\tau_j)$ for $\phi\in\lbrace 2,\ldots, 5\rbrace$ (Proposition~3 of \cite{UmsonstCDC18}).
Hence, the payoff is the detector threshold times an attacker type specific constant $\alpha_{\phi}$.
Therefore, we have that
\begin{align*}
\frac{\mathcal{I}_i^{\phi}}{\mathcal{I}_j^{\phi}}=\frac{J_D(\tau_i)}{J_D(\tau_j)}.
\end{align*}

Since this ratio is independent of the attacker type $\phi$, the closed-form solution of the defender's equilibrium MTD for each of the attacker types is the same.

\subsection{Closed-form solution}
\label{sec:SimulationsForClosedLoopSolution}
We finally evaluate the closed-form solution, where we only consider one attacker type.
The attacker that we consider uses $f(a)=\|\bar{x}\|_\infty^2$ as its attack objective, which can be seen as the attacker type with the largest payoff among the four attacker types with a non-zero payoff in the previous section.

We consider a defender that can choose from the six different mean time between false alarms in \eqref{eq:SetForTauInSim}.
Lemma~\ref{lem:DominatedStrategiesInducedMatrixGame} with $n_\phi=1$ shows us that $\tau_1=10$ is strictly dominated by the other strategies and is therefore not used in the Nash equilibrium.
Setting $i=2$, i.e., using the set of all strictly dominating strategies, we determine that $q_2=0.4748<1$ such that the condition on $i$ in Theorem~\ref{thm:ClosedFormSolutionForNashEquilibrium} holds with $i^*=2$. Using Theorem~\ref{thm:ClosedFormSolutionForNashEquilibrium}, we determine that
\begin{align}
    \label{eq:OptPSim}
    p^*=\begin{bmatrix}0 & 0.3333 & 0.1667 & 0.1 & 0.0667 & 0.3333\end{bmatrix}^T
\end{align}
and 
\begin{align}
    \label{eq:OptQSim}
   \!\!\!\!\!q^*=\begin{bmatrix}0 & 0.4748 & 0.4857 & 0.0364 & 0.0029 & 0.0002\end{bmatrix}^T.
\end{align}
These strategies are also obtained with the optimization problem \eqref{eq:LinearOptForBayesMTD}. 
So we see that our closed-form solution coincides with the solution of the linear program when $n_{\phi}=1$.
Furthermore, the MTD \eqref{eq:OptPSim} coincides with \eqref{eq:MTDScenario3}, while \eqref{eq:OptQSim} does not coincide with any of the attacker type distributions obtained in the previous section.

Next, we evaluate the effect of the size $m$ of the set of average times between false alarms the defender can choose from on the smallest index $i$ in the support set of the defender and its probability $q_i^*$.
We do so by considering values of $m$ from $1$ to $100$ and set $\tau_j=10+(j-1)50$, where $j\in[1,m]$.
The goal here is to show how the index $i$ changes as the size of the set of $\tau$ varies.
First, for a given $m$ the Nash equilibrium obtained from \eqref{eq:LinearOptForBayesMTD} coincides with the closed-form solution in Theorem~\ref{thm:ClosedFormSolutionForNashEquilibrium} for all investigated $m$.
Now let us analyse how the index $i$ and with it $q_i^*$ changes as $m$ increases.
\begin{figure}
    \centering
    \includegraphics[scale=0.5]{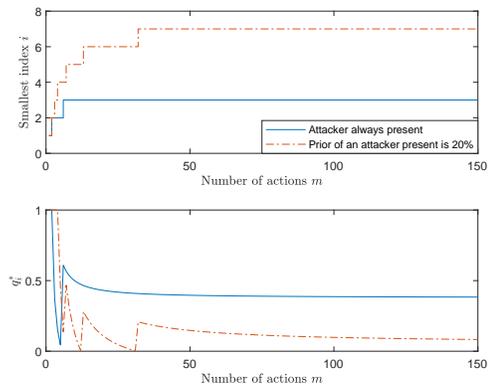}
    \caption{The upper plot shows how the smallest index $i=i^*$, for which \eqref{eq:ConditionForQi} holds, changes when $m$ increases. 
    The lower plot shows the trajectory of $q_i^*$ in \eqref{eq:ClosedFormNashForAttacker} over $m$. The solid lines represent the case with an attacker always being present, while the dashed lines represent the case where the chance of the attacker being present is $20\,\%$.}
    \label{fig:ChangeOfQi}
\end{figure}
In Figure~\ref{fig:ChangeOfQi}, we see the index $i$ of the mixed Nash equilibrium in the upper plot and $q_i^*$ in the lower plot over $m$.
First, note the evolution of $q_i^*$ over $m$. 
Furthermore, we analyse two different cases, one where the attacker is always present and one where the prior of the attacker being present is $0.2$ and the prior of an attacker with zero payoff is $0.8$ (see Remark~\ref{rem:DummyAttackerForClosedFormSolution}).
Every time $i$ changes, i.e., the smallest $\tau$ used in the Nash equilibrium changes for a given $m$, $q_i^*$ jumps from a value close to zero back up to a larger value just to decrease to zero again as $m$ increases until the next jump. 
This is consistent with $q_i^*=1-\sum_{l=i+1}^m\frac{c_{\text{F}}}{\mathcal{I}_l}\left(\frac{1}{\tau_{l-1}}-\frac{1}{\tau_{l}}\right)$, which is decreasing as $m$ increases.
Take, for example, the interval $m\in[2,6]$ for the case where there is always an attacker (solid line in Figure~\ref{fig:ChangeOfQi}). 
In this interval $i=2$ and we see that $q_i^*$ has its maximum $1$ at $m=2$ and it decays to $q_i^*=0.02456$ at $m=6$.
However, if we choose $i=2$ for $m=7$, then $q_{i}^*=-0.0148$ is smaller than zero and, therefore, using $\tau_{2}$ in the Nash equilibrium does not lead to a proper probability distribution when $m=6$.
Hence, $i$ needs to jump from $2$ to $3$ to guarantee that $q^*$ is a probability distribution with $q_i^*\in [0,1)$. 
Furthermore, if $m=100$ then only the action $\tau_{1}$ is strictly dominated.
However, in the Nash equilibrium the smallest $\tau$ used has index $i=3$. 
Therefore, we see that even if an action is not strictly dominated the action is not necessarily used in the Nash equilibrium. With our results in Theorem~\ref{thm:ClosedFormSolutionForNashEquilibrium} we understand the reasons behind that in the game presented here.
Similar observations are made for the case when the prior of the attacker with a non-zero payoff is $0.2$ (dashed lines in Figure~\ref{fig:ChangeOfQi}).
Furthermore, we see since the attacker is only present with a chance of $20\,\%$ more strategies of the defender are strictly dominated.
This observation is in line with our intuition, since the attack is less likely to happen and the defender can choose larger thresholds to avoid false alarms without fearing a larger impact.
\section{Conclusions}
\label{sec:Conclusions}
In this paper, we presented a moving target defense strategy against stealthy sensor attacks.
To find the moving target defense, we formulated a game where the defender periodically switches the detector threshold at random and the attacker has access to all sensor measurements. While the attacker wants to maximize its payoff the defender wants to minimize its cost consisting of the cost for false alarms and the cost induced by the attacker's payoff.
However, the defender is not certain about which attacker it faces and only knows the prior of the attacker's possible type.
We analyzed one period of this periodic switching game and showed that for one period we can find a strategically equivalent matrix game.
For this matrix game, we use the Bayesian Nash equilibrium to determine the equilibrium MTD strategy and presented a necessary and sufficient condition for when a MTD for the defender exists.
Furthermore, we showed that the MTD can be found by solving a linear program.
For the case, where the defender only faces one attacker/knows the attacker type exactly, we presented a closed-form solution for the moving target defense.
In the numerical evaluation, we saw how the thresholds used by the defender depend on the prior of an attack happening. 
The mere threat by the defender of randomly choosing a lower threshold with a low probability forces each attacker type to choose attacks with a lower impact which are stealthy for even small thresholds.
The reason for that is that the attacker will not be able to get a payoff when it is detected.

If we believe that the attacker might observe the defender's switching pattern before attacking, the attacker could have a larger payoff.
Therefore, one direction of future work is not to investigate the Bayesian Nash equilibrium, but the Bayesian Stackelberg equilibrium, where it is assumed that the attacker observes the defender first.
Another direction of future work would be to investigate the repeated game setting.
Furthermore, an in-depth analysis of the optimal choice of the set $\lbrace \tau_1,\ldots,\tau_m \rbrace$ is also an avenue of future work.

\bibliographystyle{IEEEtran}
\bibliography{IEEEabrv,IEEEexample}

\end{document}